\documentclass[mathpazo]{cicp}

\usepackage{color}
\usepackage{cases}
\definecolor{rouge}{rgb}{1,0,0}
\definecolor{bleu}{rgb}{0,0,1}
\definecolor{vert}{rgb}{0,0.5,0}

\begin{document}
\title[Acoustic / porous wave propagation]{Wave propagation across acoustic / Biot's media: a finite-difference method}
%Time domain numerical modeling of wave propagation in 2D acoustic / porous media

\author[Chiavassa G. and Lombard B.]
{Guillaume Chiavassa\affil{1}\comma\corrauth and 
Bruno Lombard\affil{2}}
\address{
\affilnum{1}\ Centrale Marseille, Laboratoire de M\'ecanique Mod\'elisation et Proc\'ed\'es Propres, UMR 6181 CNRS,
Technop\^ole de Chateau-Gombert, 38 rue Fr\'ed\'eric Joliot-Curie, 13451 Marseille, France. \\
\affilnum{2}\ Laboratoire de M\'{e}canique et d'Acoustique, UPR 7051 CNRS, 31 chemin Joseph Aiguier, 13402 Marseille, France.}
\emails{
{\tt guillaume.chiavassa@centrale-marseille.fr} (G.~Chiavassa), 
{\tt lombard@lma.cnrs-mrs.fr} (B.~Lombard)}

\begin{abstract}
Numerical methods are developed to simulate the wave propagation in heterogeneous 2D fluid / poroelastic media. Wave propagation is described by the usual acoustics equations (in the fluid medium) and by the low-frequency Biot's equations (in the porous medium). Interface conditions are introduced to model various hydraulic contacts between the two media: open pores, sealed pores, and imperfect pores. Well-possedness of the initial-boundary value problem is proven. Cartesian grid numerical methods previously developed in porous heterogeneous media are adapted to the present context: a fourth-order ADER scheme with Strang splitting for time-marching; a space-time mesh-refinement to capture the slow compressional wave predicted by Biot's theory; and an immersed interface method to discretize the interface conditions and to introduce a subcell resolution. Numerical experiments and comparisons with exact solutions are proposed for the three types of interface conditions, demonstrating the accuracy of the approach.
\end{abstract}
%%%%% end %%%%%%%%%%%

%%%%% AMS/PACs/Keywords %%%%%%%%%%%
%\pac{}
%35L05 wave equation
%35L50 initial boundary-value problems, 1st order hyperbolic system
%65N06 finite-difference methods
%65N85 fictitious domain methods
%74F10 linear waves

\maketitle

%------------------------------------------------------------------------------------------

\section{Introduction}\label{SecIntro}

The theory developed by Biot in 1956 \cite{BIOT56-A,BIOT56-B} is largely used to describe the wave propagation in poroelastic media. Three kinds of waves are predicted: the usual shear wave and "fast" compressional wave (as in elastodynamics), and an additional "slow" compressional wave observed experimentally in 1981 \cite{PLONA80}. This slow wave is a static mode below a critical frequency, depending on the viscosity of the saturating fluid. In the current study, we will focus on this low-frequency range.

The coupling between acoustic and poroelastic media is of high interest in many applications: sea bottom in underwater acoustics \cite{STOLL70}, borehole logging in civil engineering \cite{ROSENBAUM74}, and bones in biomechanics \cite{HAIRE99}. Many theoretical efforts have dealt with the acoustic / porous wave propagation. Various boundary conditions have been proposed to describe the hydraulic contacts: open pores, sealed pores, and imperfect pores involving the hydraulic permeability of the interface \cite{BOURBIE87,FENG83a,ROSENBAUM74}. Reflection and transmission coefficients of plane waves have been derived \cite{WU90}. The influence of the interface conditions on the existence of surface waves has been investigated in the case of inviscid \cite{FENG83a} and viscous saturating fluids \cite{EDELMAN04,GUBAIDULLIN04} in the porous material. The time-domain Green's function has been computed by the Cagniard-de Hoop's method \cite{FENG83b,DIAZ10}. Experimental works have shown the crucial importance of hydraulic contact on the generation of slow compressional wave \cite{RASOLOFOSAON88}.

The literature dedicated to numerical methods for porous wave propagation is large: see \cite{LEMOINE11}, \cite{CARCIONE10} for a review, and the introduction of \cite{CHIAVASSA11} for a list of time-domain methods. Coupled fluid / porous configurations have been addressed by an integral method \cite{GUREVICH99-B}, a spectral-element method \cite{MORENCY08}, and a pseudospectral method \cite{SIDLER10}, to cite a few. To simulate efficiently wave propagation in fluid / porous media, numerical methods must overcome the following difficulties: 
\begin{itemize}
\item in the low-frequency range, the slow compressional wave is a diffusive-like solution, and the evolution equations become stiff \cite{PUENTE08}. It drastically restricts the stability condition of  any explicit method;
\item the diffusive slow compressional wave remains localized near the interfaces. Capturing this wave - that plays a key role on the balance equations - requires a very fine spatial mesh;
\item an accurate description of arbitrary-shaped geometries with various interface conditions is crucial. These properties are badly discretized by finite-difference methods on Cartesian grids. Alternatively, unstructured meshes provide accurate descriptions, but the computational effort greatly increases.
\item an accurate modeling of the hydraulic contact at the interface is also required. In particular, as far as we know, imperfect pore conditions still have not been addressed in numerical models.
\end{itemize}
To overcome these difficulties, we adapt a methodology previously developed in porous / porous media \cite{CHIAVASSA11} and fluid / viscoelastic media \cite{LOMBARD11}. Three Cartesian grid numerical methods are put together. A fourth-order ADER scheme with Strang splitting is used to integrate the evolution equations, ensuring an optimal CFL condition of stability. Specific solvers are used in the fluid medium and in the porous medium. Their coupling is ensured by an immersed interface method, that discretizes the interface conditions and provides a subcell resolution of the geometries. Lastly, a space-time mesh refinement around the interfaces captures the small scale of the slow waves.

The article is organized as follows. In section \ref{SecPhys}, acoustics and poroelastic equations are  recalled. We introduce the interface conditions, and we prove that the initial boundary-value problem  is well-posed. In section \ref{SecNum}, numerical tools are presented. In section \ref{SecRes}, numerical experiments are proposed, based on realistic sets of physical parameters. Comparisons with analytical solutions demonstrate the accuracy of our approach. In section \ref{SecConclu}, future lines of investigation are proposed.

%------------------------------------------------------------------------------------------
%------------------------------------------------------------------------------------------

\section{Physical modeling}\label{SecPhys}

\subsection{Acoustics and Biot's equations}\label{SecPhysBiot}

\begin{figure}[htbp]
\begin{center}
\begin{tabular}{c}
\includegraphics[scale=0.70]{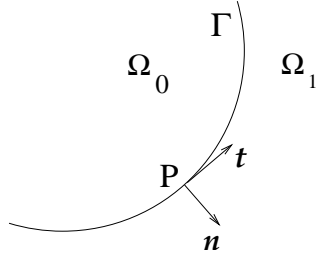} 
\end{tabular}
\end{center}
\caption{Interface $\Gamma$ separating a fluid medium $\Omega_0$ and a poroelastic medium $\Omega_1$.}
\label{Fig2D}
\end{figure}

Let us consider a 2D domain with a fluid medium $\Omega_0$ and a poroelastic medium $\Omega_1$ (figure \ref{Fig2D}). The interface $\Gamma$ separating $\Omega_0$ and $\Omega_1$ is described by a parametric equation $(x(\tau),\,y(\tau))$ (figure \ref{Fig2D}). Tangential vector ${\bf t}$ and normal vector ${\bf n}$ are defined at each point $P$ along $\Gamma$ by:
\begin{equation}
{\bf t}=(x^{'},\,y^{'})^T,\qquad {\bf n}=(y^{'},\,-x^{'})^T.
\label{NT}
\end{equation} 
The derivatives $x^{'}=\frac{d\,x}{d\,\tau}$ and $y^{'}=\frac{d\,y}{d\,\tau}$ are assumed to be continuous everywhere along $\Gamma$, and to be differentiable as many times as required further. 

In the fluid domain $\Omega_0$, the physical parameters are the density $\rho_f$ and the celerity of acoustic waves $c$. The acoustics equations write
\begin{equation}
\left\{
\begin{array}{l}
\displaystyle
\rho_f\,\frac{\textstyle \partial\,{\bf v}}{\textstyle \partial\,t}+{\bf \nabla}\,p={\bf 0},\\
[10pt]
\displaystyle
\frac{\textstyle \partial\,p}{\textstyle \partial\,t}+\rho_f\,c^2\,{\bf \nabla}.{\bf v}=f_p,
\end{array}
\right.
\label{LCfluide}
\end{equation}
where ${\bf v}=(v_1,\,v_2)^T$ is the acoustic velocity and $p$  the the acoustic pressure; $f_p$ represents an external source term.

The poroelastic medium $\Omega_1$ is modeled by the low-frequency Biot equations \cite{BOURBIE87} where the physical parameters are 
\begin{itemize}
\item the dynamic viscosity $\eta$ and the density $\rho_f$ of the saturating fluid. The latter is assumed to be the same than in $\Omega_0$, hence the notation $\rho_f$ is used in both cases; 
\item the density $\rho_s$ and the shear modulus $\mu$ of the elastic skeleton; 
\item the porosity $0<\phi<1$, the tortuosity $a>1$, the absolute permeability $\kappa$, the Lam\'e coefficient of the saturated matrix $\lambda_f$, and the two Biot's coefficients $\beta$ and $m$ of the isotropic matrix. 
\end{itemize}
The conservation of momentum and the constitutive laws yield
\begin{equation}
\left\{
\begin{array}{l}
\displaystyle
\rho\,\frac{\textstyle \partial\,{\bf v}_s}{\textstyle \partial\,t}+\rho_f\,\frac{\textstyle \partial\,{\bf w}}{\textstyle \partial\,t}-{\bf \nabla}.\boldsymbol{\sigma}={\bf 0},\\
[10pt]
\displaystyle
\rho_f\,\frac{\textstyle \partial\,{\bf v}_s}{\textstyle \partial\,t}+\rho_w\,\frac{\textstyle \partial\,{\bf w}}{\textstyle \partial\,t}+\frac{\textstyle \eta}{\textstyle \kappa}\,{\bf w}+{\nabla}\,p=0,\\
[10pt]
\displaystyle
\boldsymbol{\sigma}={\cal C}\,\boldsymbol{\varepsilon}({\bf u}_s)-\beta\,p\,{\bf I},\\
[10pt]
\displaystyle
p=-m\,\left(\beta\,{\bf \nabla}.{\bf u}_s+{\bf \nabla}.{\bf \cal W}\right),
\end{array}
\right.
\label{LCBiot}
\end{equation}
where ${\bf v}_s=\frac{\partial\,{\bf u}_s}{\partial\,t}=(v_{s1},\,v_{s2})^T$ is the elastic velocity, ${\bf w}=\phi\,({\bf v}_f-{\bf v}_s)=\frac{\partial\,{\bf \cal W}}{\partial\,t}=(w_1,\,w_2)^T$ is the filtration velocity, ${\bf v}_f$ is the fluid velocity, $\boldsymbol{\sigma}$ is the elastic stress tensor, $\boldsymbol{\varepsilon}({\bf u}_s)=\frac{1}{2}\,\left({\bf \nabla\,u}_s+^T{\bf \nabla\,u}_s\right)$ is the elastic strain tensor, and $p$ is the pressure. The following notations have also been used in (\ref{LCBiot}): $\rho_w=\frac{a}{\phi}\,\rho_f$,  $\rho=\phi\,\rho_f+(1-\phi)\,\rho_s$, and
\begin{equation}
{\cal C}=\left(
\begin{array}{ccc}
\lambda_0+2\,\mu & 0 & \lambda_0\\
[8pt]
0 & 2\,\mu & 0\\
[8pt]
\lambda_0 & 0 & \lambda_0+2\,\mu
\end{array}
\right),
\label{MatBiot}
\end{equation}
where $\lambda_0=\lambda_f-\beta^2\,m$ is the Lam\'e coefficient of the dry matrix. 

To be valid, the second equation of (\ref{LCBiot}) requires that the spectrum of the waves lies mainly in the low-frequency range, involving frequencies lower than 
\begin{equation}
f_c=\frac{\textstyle \eta\,\phi}{\textstyle 2\,\pi\,a\,\kappa\,\rho_f}.
\label{Fc}
\end{equation}
If $f\geq f_c$, more sophisticated models are required \cite{BIOT56-B,LU05} that are not addressed here. 

%------------------------------------------------------------------------------------------

\subsection{Evolution equations}\label{SecPhysEDP}

A velocity-stress formulation of the evolution equations is obtained by differentiating the last two equations in (\ref{LCBiot}) in terms of  time $t$. Setting
\begin{equation}
{\bf U}=
\left\{
\begin{array}{ll}
\displaystyle
\left(v_1,\,v_2,\,p\right)^T\qquad \hspace{3.3cm}\mbox{ in } \Omega_0,\\
[10pt]
\left(v_{s1},\,v_{s2},\,w_1,\,w_2,\,\sigma_{11},\,\sigma_{12},\,\sigma_{22},\,p\right)^T\qquad \mbox{ in } \Omega_1,
\end{array}  
\right.
\label{VecU}
\end{equation}
where $\sigma_{11}$, $\sigma_{12}$, and $\sigma_{22}$ are the independent components of the stress tensor $\boldsymbol{\sigma}$, one deduces from (\ref{LCfluide}) and (\ref{LCBiot}) the first-order linear system of partial differential equations with source term
\begin{equation}
\frac{\textstyle \partial}{\textstyle \partial\,t}\,{\bf U}+{\bf A}\,\frac{\textstyle \partial}{\textstyle \partial\,x}\,{\bf U}+{\bf B}\,\frac{\textstyle \partial}{\textstyle \partial\,y}\,{\bf U}=-{\bf S}\,{\bf U}+{\bf F}.
\label{SystHyp}
\end{equation}
The system (\ref{SystHyp}) is completed by initial values and radiation conditions at infinity.

In (\ref{SystHyp}), ${\bf A}$, ${\bf B}$ and ${\bf S}$ are $3\times3$ matrices in $\Omega_0$, and $8\times8$ matrices in $\Omega_1$; the vector ${\bf F}$ accounts for the acoustic source in (\ref{LCfluide}). In $\Omega_0$, ${\bf S}={\bf 0}$, while in $\Omega_1$ the spectral radius of ${\bf S}$ is
\begin{equation}
R({\bf S})=\frac{\textstyle \eta}{\textstyle \kappa}\,\frac{\textstyle \rho}{\textstyle \rho\,\rho_w-\rho_f^2},
\label{RadiusSpectrum}
\end{equation} 
which can be large, depending on the hydraulic permeability $\eta/ \kappa$.

\begin{figure}[htbp]
\begin{center}
\begin{tabular}{cc}
(a) & (b)\\
\includegraphics[scale=0.33]{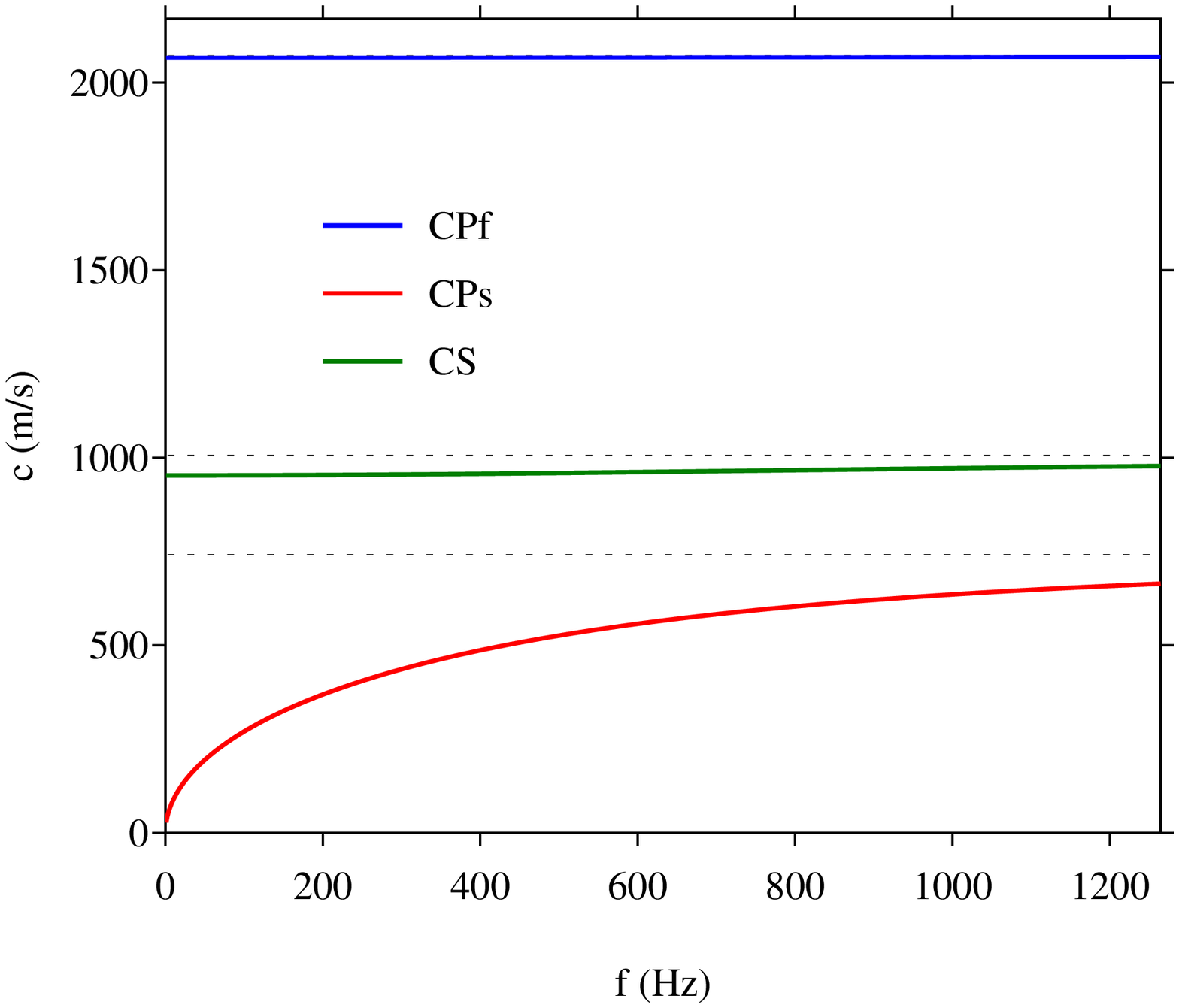} &
\includegraphics[scale=0.33]{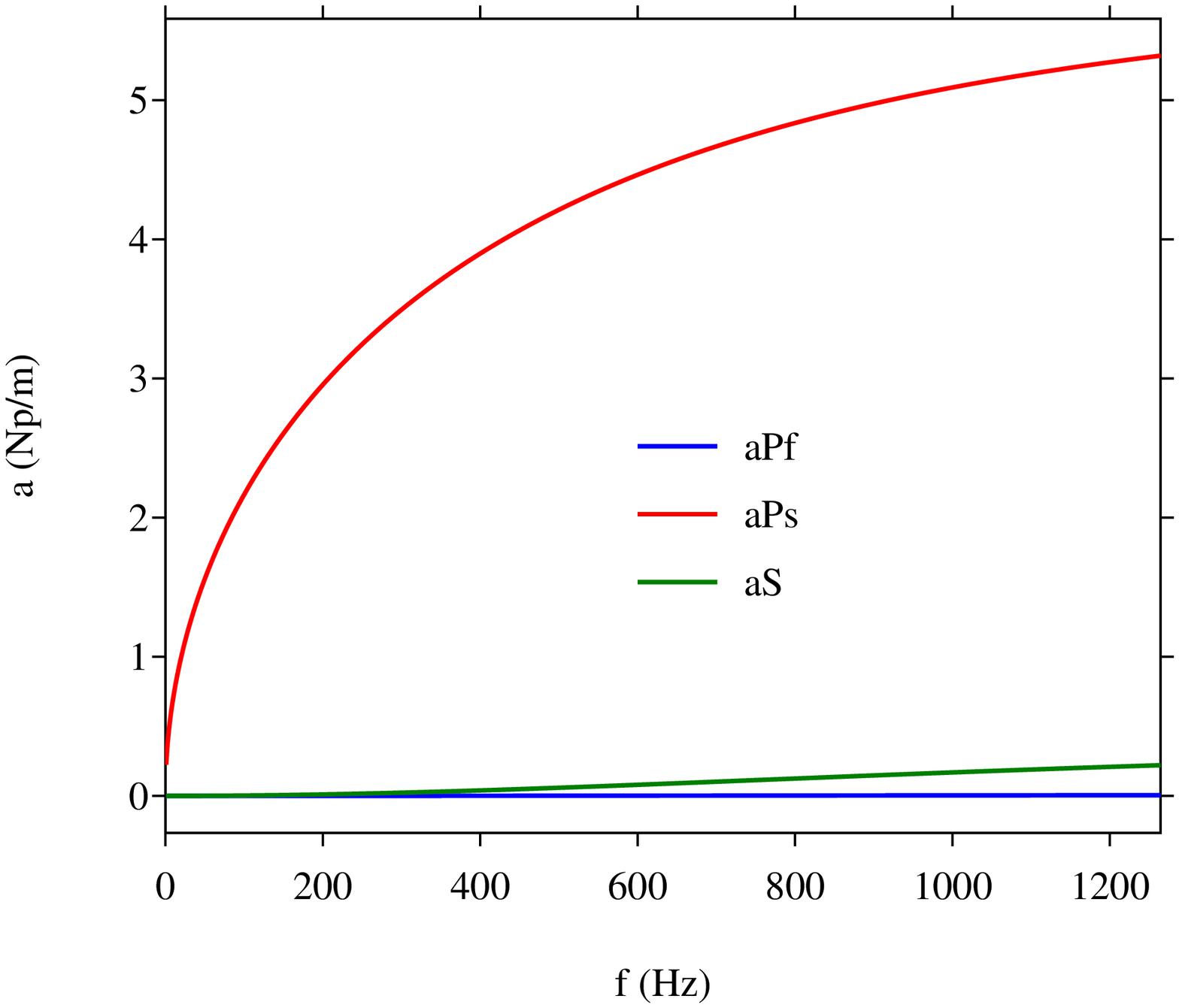} 
\end{tabular}
\end{center}
\caption{Phase velocities (a) and attenuations (b) of Biot's waves, using the parameters given in table \ref{TabParametres}. $pf$: fast compressional wave; $ps$: slow compressional wave; $s$: shear wave. In (a), the horizontal dotted lines refer to the eigenvalues $c_{pf}^\infty$, $c_{ps}^\infty$ and $c_{s}^\infty$ of ${\bf A}$ and ${\bf B}$.}
\label{FigDispersion}
\end{figure}

The non-null eigenvalues of ${\bf A}$ and ${\bf B}$ are real: $\pm c$ (acoustic wave) in $\Omega_0$; $\pm c_{pf}^\infty$ (fast compressional wave), $\pm c_{ps}^\infty$ (slow compressional wave), and $\pm c_{s}^\infty$ (shear wave) in $\Omega_1$, satisfying $0<\max(c_{ps}^\infty,\,c_{s}^\infty)<c_{pf}^\infty$. These eigenvalues in $\Omega_1$ are the high-frequency limits of the phase velocities of the poroelastic waves:
\begin{equation}
\lim_{f\rightarrow +\infty}c_{pf}(f)=c_{pf}^\infty,\qquad
\lim_{f\rightarrow +\infty}c_{ps}(f)=c_{ps}^\infty,\qquad
\lim_{f\rightarrow +\infty}c_{s}(f)=c_{s}^\infty.
\label{Dispersion}
\end{equation}
The fast compressional wave and the shear wave are almost non-dispersive and non-diffusive solutions. On the contrary, the phase velocity of the slow compressional wave tends to zero with frequency (figure \ref{FigDispersion}-a); at higher frequencies, this wave propagates, but it is highly attenuated (figure \ref{FigDispersion}-b). For a detailed dispersion analysis, the reader is referred to standard texts \cite{BIOT56-A,CARCIONE07}.

%------------------------------------------------------------------------------------------

\subsection{Hydraulic contacts}\label{SecPhysJC}

Four waves are involved in the acoustic / porous configuration: one acoustic wave in $\Omega_0$, and three poroelastic waves in $\Omega_1$. Consequently, four independent interface conditions need to be defined along $\Gamma$. Authors \cite{BOURBIE87,FENG83a,GUREVICH99,ROSENBAUM74} have proposed the following general conditions:
\begin{subnumcases}{\label{JC}}
\displaystyle
{\bf v}_0.{\bf n}={\bf v}_{s1}.{\bf n}+{\bf w}_1.{\bf n},\label{JCa}\\
[10pt]
\displaystyle
-p_0\,{\bf n}=\boldsymbol{\sigma}_1.{\bf n},\label{JCb}\\
[10pt]
\displaystyle
[p]=-\frac{\textstyle 1}{\textstyle \mathcal{K}}\,\frac{\textstyle {\bf w}_1.{\bf n}}{\textstyle |{\bf n}|},\label{JCc}
\end{subnumcases}
where the subscripts 0 and 1 refer to the traces on the $\Omega_0$ or $\Omega_1$ sides, and $[p]$ denotes the jump of $p$ from $\Omega_0$ to $\Omega_1$. The first scalar equation (\ref{JCa}) follows from the conservation of fluid mass. The second vectorial equation (\ref{JCb}) expresses the continuity of normal efforts. The third scalar equation (\ref{JCc}) is a local Darcy's law. It models the hydraulic contact between the fluid and the porous medium, and  involves an additional parameter $\mathcal{K}$, called the hydraulic permeability of the interface. The division by $|{\bf n}|$ ensures that the hydraulic contact is independent from the choice of the parametric equation of $\Gamma$. According to the value of $\mathcal{K}$, various limit-cases are encountered:
\begin{itemize}
\item if $\mathcal{K} \rightarrow +\infty$, then the last equation in (\ref{JC}) becomes $[p]=0$, modeling the commonly used {\it open pores};
\item if $\mathcal{K}\rightarrow 0$, then no fluid exchange occurs across $\Gamma$, and the last equation in (\ref{JC}) is replaced by ${\bf w}_1.{\bf n}= 0$, modeling {\it sealed pores};  
\item if $0<\mathcal{K}<+\infty$, then an intermediate state between open pores and sealed pores is reached, modeling {\it imperfect pores}.
\end{itemize}

The following proposition states that the interface conditions (\ref{JC}) coupled with the evolution equations (\ref{LCfluide}) and (\ref{LCBiot}) yield a well-posed problem.

\begin{proposition}
Let
$$
E=E_1+E_2+E_3,
$$
with
\begin{equation}
\begin{array}{lll}
E_1 &=& \displaystyle \frac{\textstyle 1}{\textstyle 2}\int_{\Omega_0}\left(\rho_f\,{\bf v}^2+\frac{\textstyle 1}{\textstyle \rho_f\,c^2}\,p^2\right)\,d\Omega,\\
\\
E_2 &=& \displaystyle \frac{\textstyle 1}{\textstyle 2}\int_{\Omega_1}\left(\rho\,{\bf v}_s^2+\rho_w\,{\bf w}^2+2\,\rho_f\,{\bf v}_s\,{\bf w}\right)\,d\Omega,\\
\\
E_3 &=& \displaystyle \frac{\textstyle 1}{\textstyle 2}\int_{\Omega_1}\left({\cal C}\,\boldsymbol{\varepsilon}({\bf u}_s):\boldsymbol{\varepsilon}({\bf u}_s)+\frac{\textstyle 1}{\textstyle m}\,p^2\right)\,d\Omega.
\end{array}
\label{Energie}
\end{equation}
Then, $E$ is an energy which satisfies 
\begin{equation}
\frac{\textstyle d\,E}{\textstyle d\,t}=-\int_{\Omega_1}\frac{\textstyle \eta}{\textstyle \kappa}\,{\bf w}^2\,d\Omega-\int_{\Gamma}\frac{\textstyle 1}{\textstyle \mathcal{K}}\,\frac{\textstyle \left({\bf w}_1.{\bf n}\right)^2}{\textstyle |{\bf n}|}\,d\Gamma \, \leq 0.
\label{dEdt}
\end{equation}
\label{PropositionNRJ}
\end{proposition}

\begin{proof} The first equation of (\ref{LCfluide}) is multiplied by ${\bf v}$ and  integrated over $\Omega_0$:
\begin{equation}
\int_{\Omega_0}\left(\rho_f\,{\bf v}\,\frac{\textstyle \partial\,{\bf v}}{\textstyle \partial\,t}+{\bf v}\,{\bf \nabla}\,p\right)\,d\Omega=0.
\label{Proof1}
\end{equation}
The first term in (\ref{Proof1}) writes
\begin{equation}
\int_{\Omega_0}\rho_f\,{\bf v}\,\frac{\textstyle \partial\,{\bf v}}{\textstyle \partial\,t}\,d\Omega=\frac{\textstyle d}{\textstyle d\,t}\,\frac{\textstyle 1}{\textstyle 2}\int_{\Omega_0}\rho_f\,{\bf v}^2\,d\Omega.
\label{Proof2}
\end{equation}
By integration by part, and using the second equation of (\ref{LCfluide}), one obtains
\begin{equation}
\begin{array}{lll}
\displaystyle
\int_{\Omega_0}{\bf v}\,{\bf \nabla}\,p\,d\Omega&=& \displaystyle \int_{\Gamma}{\bf v}_0.{\bf n}\,p_0\,d\Gamma-\int_{\Omega_0}{\bf \nabla}.{\bf v}\,p\,d\Omega,\\
[12pt]
&=&\displaystyle \int_{\Gamma}{\bf v}_0.{\bf n}\,p_0\,d\Gamma+\int_{\Omega_0}\frac{\textstyle 1}{\textstyle \rho_f\,c^2}\,p\,\frac{\textstyle \partial\,p}{\textstyle \partial\,t}\,d\Omega,\\
[12pt]
&=&\displaystyle \int_{\Gamma}{\bf v}_0.{\bf n}\,p_0\,d\Gamma+\frac{\textstyle d}{\textstyle d\,t}\,\frac{\textstyle 1}{\textstyle 2}\int_{\Omega_0}\frac{\textstyle 1}{\textstyle \rho_f\,c^2}\,p^2\,d\Omega,
\end{array}
\label{Proof3}
\end{equation}
which concludes the energy analysis in the fluid domain. The first equation of (\ref{LCBiot}) is multiplied by ${\bf v}_s$ and is integrated over $\Omega_1$:
\begin{equation}
\int_{\Omega_1}\left(\rho\,{\bf v}_s\,\frac{\textstyle \partial\,{\bf v}_s}{\textstyle \partial\,t}+\rho_f\,{\bf v}_s\,\frac{\textstyle \partial\,{\bf w}}{\textstyle \partial\,t}-{\bf v}_s\,{\bf \nabla}.\boldsymbol{\sigma}\right)\,d\Omega=0.
\label{Proof4}
\end{equation}
The first term in (\ref{Proof4}) writes
\begin{equation}
\int_{\Omega_1}\rho\,{\bf v}_s\,\frac{\textstyle \partial\,{\bf v}_s}{\textstyle \partial\,t}\,d\Omega=\frac{\textstyle d}{\textstyle d\,t}\,\frac{\textstyle 1}{\textstyle 2}\int_{\Omega_1}\rho\,{\bf v}_s^2\,d\Omega.
\label{Proof5}
\end{equation}
By integration by part, and using the third equation of (\ref{LCBiot}), one obtains
\begin{equation}
\begin{array}{lll}
\displaystyle
-\int_{\Omega_1}{\bf v}_s\,{\bf \nabla}.\boldsymbol{\sigma}\,d\Omega &=& \displaystyle \int_{\Gamma}{\bf v}_{s1}.\boldsymbol{\sigma}_1.{\bf n}\,d\Gamma+\int_{\Omega_1}\boldsymbol{\varepsilon}({\bf v}_s):\boldsymbol{\sigma}\,d\Omega,\\
[12pt]
&=& \displaystyle \int_{\Gamma}{\bf v}_{s1}.\boldsymbol{\sigma}_1.{\bf n}\,d\Gamma+\int_{\Omega_1}\boldsymbol{\varepsilon}({\bf v}_s):\left({\cal C}\,\boldsymbol{\varepsilon}({\bf u}_s)-\beta\,p\,{\bf I}\right)\,d\Omega,\\
[12pt]
&=& \displaystyle \int_{\Gamma}{\bf v}_{s1}.\boldsymbol{\sigma}_1.{\bf n}\,d\Gamma+\int_{\Omega_1}  \frac{\textstyle \partial}{\textstyle \partial\,t}\,\boldsymbol{ \varepsilon}({\bf u}_s):{\cal C}\,\boldsymbol{\varepsilon}({\bf u}_s)\,d\Omega-\int_{\Omega_1}\beta\,p\,{\bf \nabla}.{\bf v}_s\,d\Omega,\\
[12pt]
&=& \displaystyle \int_{\Gamma}{\bf v}_{s1}.\boldsymbol{\sigma}_1.{\bf n}\,d\Gamma+\frac{\textstyle d}{\textstyle d\,t}\,\left(\frac{\textstyle 1}{\textstyle 2}\,\int_{\Omega_1}{\cal C}\,\boldsymbol{\varepsilon}({\bf u}_s):\boldsymbol{\varepsilon}({\bf u}_s)\,d\Omega\right)-\int_{\Omega_1}\beta\,p\,{\bf \nabla}.{\bf v}_s\,d\Omega.
\end{array}
\label{Proof6}
\end{equation}
The second equation of (\ref{LCBiot}) is multiplied by ${\bf w}$, and then it is integrated over $\Omega_1$:
\begin{equation}
\int_{\Omega_1}\left(\rho_f\,{\bf w}\,\frac{\textstyle \partial\,{\bf v}_s}{\textstyle \partial\,t}+\rho_w\,{\bf w}\,\frac{\textstyle \partial\,{\bf w}}{\textstyle \partial\,t}+\frac{\textstyle \eta}{\textstyle \kappa}\,{\bf w}^2+{\bf w}\,{\bf \nabla}\,p\right)\,d\Omega=0.
\label{Proof7}
\end{equation}
The second term in (\ref{Proof7}) writes
\begin{equation}
\int_{\Omega_1}\rho_w\,{\bf w}\,\frac{\textstyle \partial\,{\bf w}}{\textstyle \partial\,t}\,d\Omega=\frac{\textstyle d}{\textstyle d\,t}\,\frac{\textstyle 1}{\textstyle 2}\int_{\Omega_1}\rho_w\,{\bf w}^2\,d\Omega.
\label{Proof8}
\end{equation}
By integration by part, and using the fourth equation of (\ref{LCBiot}), one obtains
\begin{equation}
\begin{array}{lll}
\displaystyle
\int_{\Omega_1}{\bf w}\,{\bf \nabla}\,p\,d\Omega&=& \displaystyle -\int_{\Gamma}{\bf w}_1.{\bf n}\,p_1\,d\Gamma-\int_{\Omega_1}p\,{\bf \nabla}.{\bf w}\,d\Omega,\\
[12pt]
&=& \displaystyle -\int_{\Gamma}{\bf w}_1.{\bf n}\,p_1\,d\Gamma+\int_{\Omega_1}p\,\frac{\textstyle \partial}{\textstyle \partial\,t}\,\left(\frac{\textstyle 1}{\textstyle m}\,p+\beta\,{\bf \nabla}.{\bf u}_s \right)\,d\Omega,\\
[12pt]
&=& \displaystyle -\int_{\Gamma}{\bf w}_1.{\bf n}\,p_1\,d\Gamma+\frac{\textstyle d}{\textstyle d\,t}\,\left(\frac{\textstyle 1}{\textstyle 2}\,\int_{\Omega_1}\frac{\textstyle 1}{\textstyle m}\,p^2\,d\Omega\right)+\int_{\Omega_1}\beta\,p\,{\bf \nabla}.{\bf v}_s\,d\Omega.
\end{array}
\label{Proof9}
\end{equation}
When adding (\ref{Proof4}) and (\ref{Proof7}), it remains  
\begin{equation}
\int_{\Omega_1}\rho_f\,\left({\bf v}_s\,\frac{\textstyle \partial\,{\bf w}}{\textstyle \partial\,t}+{\bf w}\,\frac{\textstyle \partial\,{\bf v}_s}{\textstyle \partial\,t}\right)\,d\Omega=\frac{\textstyle d}{\textstyle d\,t}\int_{\Omega_1}\rho_f\,{\bf v}_s\,{\bf w}\,d\Omega.
\label{Proof10}
\end{equation}
Equations (\ref{Proof1})-(\ref{Proof10}) provide
\begin{equation}
\frac{\textstyle d\,E}{\textstyle d\,t}=-\int_{\Omega_1}\frac{\textstyle \eta}{\textstyle \kappa}\,{\bf w}^2\,d\Omega-\psi,
\label{Proof11}
\end{equation}
where $\psi$ is simplified from the conditions (\ref{JC}):
\begin{equation}
\begin{array}{lll}
\displaystyle \psi &=& \displaystyle \int_{\Gamma}\left({\bf v}_0.{\bf n}\,p_0+{\bf v}_{s1}.{\bf n}\,\boldsymbol{\sigma}_1\,{\bf n}-{\bf w}_1.{\bf n}\,p_1\right)\,d\Gamma,\\
[12pt]
&=& \displaystyle \int_{\Gamma}\left({\bf v}_0.{\bf n}\,p_0-{\bf v}_{s1}.{\bf n}\,p_0-{\bf w}_1.{\bf n}\,p_1\right)\,d\Gamma,\\
[12pt]
&=& \displaystyle \int_{\Gamma}\left(\left({\bf v}_{s1}.{\bf n}+{\bf w}_1.{\bf n}\right)\,p_0-{\bf v}_{s1}.{\bf n}\,p_0-{\bf w}_1.{\bf n}\,p_1\right)\,d\Gamma,\\
[12pt]
&=& \displaystyle -\int_{\Gamma}{\bf w}_1.{\bf n}\,[p]\,d\Gamma,\\
[12pt]
&=& \displaystyle +\int_{\Gamma}\frac{\textstyle 1}{\textstyle \mathcal{K}}\,\frac{\textstyle \left({\bf w}_1.{\bf n}\right)^2}{\textstyle |{\bf n}|}\,d\Gamma,
\end{array}
\label{Proof12}
\end{equation}
which gives (\ref{dEdt}). It remains to prove that $E$ is a positive definite quadratic form. This is obviously the case of $E_1$ and $\frac{1}{m}\,p^2$ in $E_3$. Definite positivity of ${\cal C}\,\boldsymbol{\varepsilon}({\bf u}_s):\boldsymbol{\varepsilon}({\bf u}_s)$ is assumed by hypothesis on the elastic tensor ${\cal C}$ \cite{THESE_EZZIANI}, which concludes the analysis of $E_3$. Lastly, the matrix
$$
{\bf M}=
\left(
\begin{array}{cc}
\rho & \rho_f\\
\rho_f & \rho_w
\end{array}
\right)
$$
is definite positive:
$$
\begin{array}{lll}
\det\,{\bf M}&=&\rho\,\rho_w-\rho_f^2,\\
[6pt]
&=&\displaystyle (a-1)\,\rho_f^2+\frac{\textstyle 1-\phi}{\textstyle \phi}\,a\,\rho_f\,\rho_s\geq 0,
\end{array}
$$
which concludes the analysis of $E_2$.
\end{proof} 

Each term in (\ref{Energie}) has a clear physical significance: $E_1$ is the acoustical energy, $E_2$ is the poroelastic kinetic energy, and $E_3$ is the poroelastic potential energy; $E_3$ is easily computed from (\ref{MatBiot}) using the closed-form expression:
\begin{equation}
\begin{array}{lll}
\displaystyle
{\cal C}\,\boldsymbol{\varepsilon}({\bf u}_s):\boldsymbol{\varepsilon}({\bf u}_s)&=&
\displaystyle\frac{\textstyle \lambda_0+2\,\mu}{\textstyle 4\,\mu\,\left(\lambda_0+\mu\right)}\,\left((\sigma_{11}+\beta\,p)^2+(\sigma_{22}+\beta\,p)^2\right)+\frac{\textstyle 1}{\textstyle \mu}\,\sigma_{12}^2\\
[10pt]
&-&\displaystyle \frac{\textstyle \lambda_0}{\textstyle 2\,\mu\,\left(\lambda_0+\mu\right)}\,(\sigma_{11}+\beta\,p)\,(\sigma_{22}+\beta\,p).
\end{array}
\end{equation}

The decrease rate of the total energy is governed as usual by the intrinsic attenuation due to the viscous saturating fluid $\int_{\Omega_1}\frac{\eta}{\kappa}\,{\bf w}^2\,d\Omega$, but also by the imperfect pore condition $\int_{\Gamma}\frac{1}{\mathcal{K}}\,\left({\bf w}_1.{\bf n}\right)^2\,d\Gamma$. In particular, even if the viscous effects are neglected ($\eta=0$), a part of the mechanical energy is dissipated by the interface in the case of imperfect pores.

%------------------------------------------------------------------------------------------

\subsection{Matrix formulation of the interface conditions}\label{SecPhysMatrix}

The various limit-cases in (\ref{JC}) are written in an alternative way, well-suited for further discretization of interface conditions (section \ref{SecNumIIM} and appendix \ref{SecDetailsIIM}). First, the open pore conditions yield
\begin{equation}
\left\{
\begin{array}{l}
\displaystyle
{\bf v}_0.{\bf n}={\bf v}_{s1}.{\bf n}+{\bf w}_1.{\bf n},\\
[10pt]
\displaystyle
-p_0\,{\bf n}^2=\left(\boldsymbol{\sigma}_1.{\bf n}\right).\,{\bf n},\\
[10pt]
\displaystyle
\left(\boldsymbol{\sigma}_1.{\bf n}\right).\,{\bf t}=0,\\
[10pt]
\displaystyle
\left(\boldsymbol{\sigma}_1.{\bf n}\right).\,{\bf n}+p_1\,{\bf n}^2=0.
\end{array}
\right.
\label{JCOpen}
\end{equation}
Second, the sealed pore conditions yield
\begin{equation} 
\left\{
\begin{array}{l}
\displaystyle
{\bf v}_0.{\bf n}={\bf v}_{s1}.{\bf n},\\
[10pt]
\displaystyle
-p_0\,{\bf n}^2=\left(\boldsymbol{\sigma}_1.{\bf n}\right).\,{\bf n},\\
[10pt]
\displaystyle
\left(\boldsymbol{\sigma}_1.{\bf n}\right).\,{\bf t}=0,\\
[10pt]
\displaystyle
{\bf w}_1.{\bf n}=0.
\end{array}
\right.
\label{JCClose}
\end{equation}
Third and last, the imperfect pore conditions yield
\begin{equation}
\left\{
\begin{array}{l}
\displaystyle
{\bf v}_0.{\bf n}={\bf v}_{s1}.{\bf n}+{\bf w}_1.{\bf n},\\
[10pt]
\displaystyle
-p_0\,{\bf n}^2=\left(\boldsymbol{\sigma}_1.{\bf n}\right).\,{\bf n},\\
[10pt]
\displaystyle
\left(\boldsymbol{\sigma}_1.{\bf n}\right).\,{\bf t}=0,\\
[10pt]
\displaystyle
\left(\boldsymbol{\sigma}_1.{\bf n}\right).\,{\bf n}+p_1\,{\bf n}^2+\frac{\textstyle 1}{\textstyle \mathcal{K}}\,{\bf w}_1.{\bf n}|{\bf n}|=0.
\end{array}
\right.
\label{JCPerm}
\end{equation}
In the three cases (\ref{JCOpen})-(\ref{JCPerm}), the four scalar interface conditions are written as two jump conditions and two boundary conditions on the $\Omega_1$ side. 

These conditions are now written in a matrix way, useful for the further derivation of $r$-th order interface conditions ($r\geq 1$). On the side $\Omega_i$ ($i=0,\,1$), the boundary values of the spatial derivatives of ${\bf U}$ up to the $r$-th order are put in a vector ${\bf U}^r_i$:
\begin{equation}
\begin{array}{l}
\displaystyle
{\bf U}^r_i=\lim_{M\rightarrow P,\,M\in\Omega_i}
\left(
{\bf U}^T,
...,\,
\frac{\textstyle \partial^l}{\textstyle \partial\, x^{l-m}\,\partial\,y^m}\,{\bf U}^T,
...,\,
\frac{\textstyle \partial^r}{\textstyle \partial\,y^r}\,{\bf U}^T
\right)^T,
\label{Ur}
\end{array}
\end{equation}
where $l=0,\,...,\,r$ and $m=0,\,...,\,l$. The vector ${\bf U}^r_i$ has $n_v=3\,(r+1)\,(r+2)/2$ components in the fluid domain $\Omega_0$, and $n_v=4\,(r+1)\,(r+2)$ components in the poroelastic domain $\Omega_1$. Based on this formalism, the zero-th order interface conditions (\ref{JCOpen})-(\ref{JCPerm}) are written 
\begin{equation}
{\bf C}_1^0\,{\bf U}_1^0={\bf C}_0^0\,{\bf U}_0^0,\qquad
{\bf L}_1^0\,{\bf U}_1^0={\bf 0},
\label{JC0}
\end{equation}
where ${\bf C}_0^0$ is a $2 \times 3$ matrix, and ${\bf C}_1^0$, ${\bf L}_1^0$ are $2 \times 8$ matrices; they are detailed in appendix \ref{SecAnnexeMat}.  

%------------------------------------------------------------------------------------------
%------------------------------------------------------------------------------------------

\section{Numerical modeling}\label{SecNum}

\subsection{Integration of evolution equations}\label{SecNumLC}

The system (\ref{SystHyp}) is solved on a uniform Cartesian grid, with spatial mesh sizes $\Delta\,x=\Delta\,y$ and a time step $\Delta\,t$. Due to the source term in the poroelastic medium $\Omega_1$, a straightforward discretization of (\ref{SystHyp}) is inefficient: a Von-Neumann analysis of stability gives
\begin{equation}
\Delta t\leq \min \left(\frac{\textstyle \Delta x}{\textstyle \max(c)},\,\frac{\textstyle 2}{\textstyle R({\bf S})}\right),
\label{CFLdirect}
\end{equation}
where the spectral radius $R({\bf S})$ can become large (\ref{RadiusSpectrum}). It is more efficient to split (\ref{SystHyp}) and to solve alternatively
\begin{subnumcases}{\label{Strang}}
\displaystyle
\frac{\textstyle \partial}{\textstyle \partial\,t}\,{\bf U}+{\bf A}\,\frac{\textstyle \partial}{\textstyle \partial\,x}\,{\bf U} + {\bf B}\,\frac{\textstyle \partial}{\textstyle \partial\,y}\,{\bf U}={\bf 0}, \label{Strang1}\\
%[6pt]
\displaystyle
\frac{\textstyle \partial}{\textstyle \partial\,t}\,{\bf U}=-{\bf S}\,{\bf U}.
\hspace{3.1cm} \label{Strang2}
\end{subnumcases}
The operators used to solve the propagative part (\ref{Strang1}) and the diffusive part (\ref{Strang2}) are denoted by ${\bf H}_a$ and ${\bf H}_b$, respectively. Second-order Strang splitting is used \cite{LEVEQUE02}, hence time-marching can be written
\begin{equation}
\begin{array}{lllll}
\displaystyle
&\bullet& {\bf U}_{i,j}^{(1)}&=&{\bf H}_{b}(\frac{\Delta\,t}{2})\,{\bf U}_{i,j}^{n},\\
[6pt]
\displaystyle
&\bullet& {\bf U}_{i,j}^{(2)}&=&{\bf H}_{a}(\Delta\,t)\,{\bf U}_{i,j}^{(1)},\\
[6pt]
\displaystyle
&\bullet& {\bf U}_{i,j}^{n+1}&=&{\bf H}_{b}(\frac{\Delta\,t}{2})\,{\bf U}_{i,j}^{(2)}.
\end{array}
\label{AlgoSplitting}
\end{equation}
The discrete operator ${\bf H}_a$ in (\ref{AlgoSplitting}) is a fourth-order ADER scheme \cite{KASER07}. On a Cartesian grid, this  explicit scheme amounts to a fourth-order Lax-Wendroff scheme. It satisfies the stability condition $\max (c)\,\frac{\Delta\,t}{\Delta\,x}\leq 1$. The diffusive operator is exact and unconditionally stable: ${\bf H}_b(t)=e^{-{\bf S}\,t}$, see equation (18) in \cite{CHIAVASSA11}. If dissipation processes are neglected, then ${\bf H}_b$ is also the identity: it is always the case in $\Omega_0$, and it is also true in $\Omega_1$ when $\eta=0$.

When ${\bf S}\neq{\bf 0}$ (i.e. when $\eta \neq 0$), the coupling between parts (\ref{Strang1}) and (\ref{Strang2}) decreases formally the convergence rate from 4 to 2. In counterpart, the optimal condition of stability is recovered: $\max (c)\,\frac{\Delta\,t}{\Delta\,x}\leq 1$, which is essential for real parameters simulations. Without splitting, the time step would be divided by about 3 and 4 in the Test 2 and Test 3 of section \ref{SecRes}, respectively. 

%------------------------------------------------------------------------------------------

\subsection{Mesh refinement}\label{SecNumAMR}

In the low-frequency range, the slow compressional wave behaves like a diffusive non-propagating solution, with very small wavelength (section \ref{SecPhysEDP}). A very fine spatial mesh is therefore required. Since this wave remains localized at the interfaces where it is generated, space-time mesh refinement is a good strategy. For this purpose, we adapt a steady-state version of the algorithm proposed in \cite{BERGER98,BERGER84}. In the fine grid, both the spatial meshes and the time step are divided by a refinement factor $q$. Doing so ensures the same CFL number in each grid. The coupling between coarse and fine meshes is based on spatial and temporal interpolations.

Even if the refined zone is much smaller than the whole domain, the refinement greatly increases the computational cost. The factor $q$ therefore must be chosen adequately. We choose $q$ that ensures the same discretization of the slow wave, on the refined zone, than the fast wave, on the coarse grid. Direct calculations give
\begin{equation}
q(f_0)=\frac{{c}_{pf}(f_0)}{c_{ps}(f_0)},
\label{AMR}
\end{equation}
where $f_0$ is the central frequency of the signal. 

%------------------------------------------------------------------------------------------

\subsection{Discretization of the interface conditions}\label{SecNumIIM}

The discretization of the interface conditions requires special care. A straightforward stair-step representation of interfaces introduces first-order geometrical errors and yields spurious numerical diffractions. In addition, the jump conditions (\ref{JC}) are not enforced numerically if no special treatment is applied. Lastly, the smoothness requirements to solve (\ref{Strang1}) are not satisfied, decreasing the convergence rate of the ADER scheme. 

To remove these drawbacks while maintaining the efficiency of Cartesian grid methods, immersed interface methods constitute a good strategy \cite{LI94,LEVEQUE97}. Various formulations have been proposed in the literature; here, we follow the methodology proposed in acoustics / elastic media \cite{LOMBARD04}, viscoelastic media \cite{LOMBARD11}, and poroelastic media \cite{CHIAVASSA11}. 

\begin{figure}[htbp]
\begin{center}
\includegraphics[scale=0.75]{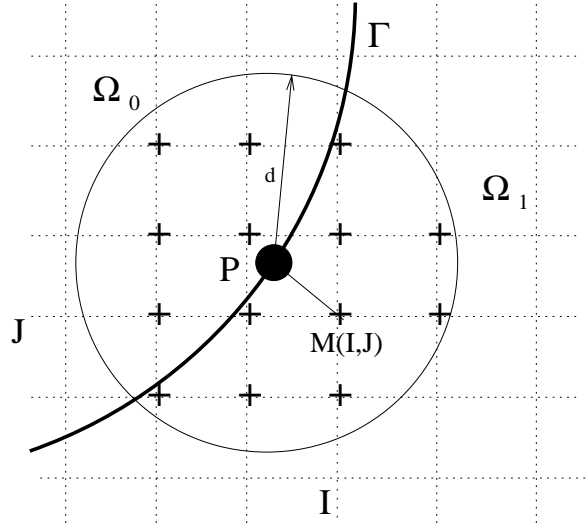} 
\caption{Irregular point $M(x_I,\,y_J)\in \Omega_1$ and its orthogonal projection $P$ onto $\Gamma$. The grid nodes used to compute ${\bf U}_{I,J}^*$ are inside the circle with radius $d$ and centered on $P$; they are denoted by ${\bf +}$.}
\label{Patate}
\end{center}
\end{figure}

To illustrate the basic principle of this method, let us take an irregular point $(x_i,\,y_j)\in\Omega_0$, where time-marching (\ref{Strang1}) uses a value at $(x_I,\,y_J)\in\Omega_1$ (figure \ref{Patate}). Instead of using ${\bf U}_{I,J}^{(1)}$, the discrete operator ${\bf H}_a$ uses a modified value ${\bf U}_{I,J}^*$. This latter is a $r$-th order extension of the solution from $\Omega_0$ into $\Omega_1$.

In a few words, ${\bf U}_{I,J}^*$ is build as follows. Let $P$ be the orthogonal projection of $(x_I,\,y_J)$ on $\Gamma$, and consider the disc ${\cal D}$ centered on $P$ with a radius $d$ (figure \ref{Patate}). Based on the interface conditions (\ref{JC0}) at $P$ and on the numerical values ${\bf U}^{(1)}$ at the grid nodes inside ${\cal D}$, a matrix ${\cal M}$ is build so that
\begin{equation}
{\bf U}_{I,J}^*={\cal M}\left({\bf U}^{(1)}\right)_{\mathcal D}.
\label{UIJ*}
\end{equation}
The derivation of matrix ${\cal M}$ in (\ref{UIJ*}) is detailed in appendix \ref{SecDetailsIIM}. Some comments are done about the immersed interface method:
\begin{enumerate}
\item A similar algorithm is applied at each irregular point along $\Gamma$ and at each propagative part of the splitting algorithm (\ref{Strang1}). Since the jump conditions do not vary with time, the evaluation of the matrices in (\ref{UIJ*}) is done during a preprocessing step. Only small matrix-vector products are therefore required at each splitting step. After optimization of the computer codes, this additional cost is made negligible, lower than 1\% of the time-marching.

\item The matrix ${\cal M}$ in (\ref{UIJ*}) depends on the subcell position of $P$ inside the mesh and on the jump conditions at $P$, involving the local geometry and the curvature of $\Gamma$ at $P$. Consequently, all these insights are incorporated in the modified value (\ref{UIJ*}), and hence in the scheme.

\item \label{EstimeVarEps} The simulations indicate that the number of grid nodes inside the disc ${\cal D}$ has a crucial influence on the stability of the immersed interface method. Here we use a constant radius $d$. Taking $r=2$, numerical experiments have shown that $d=3.2\,\Delta\,x$ is a good candidate, while $d=4.5\,\Delta\,x$ is used when $r=3$. 

\item The order $r$ plays an important role on the accuracy of the coupling between the immersed interface method and a $k$-th order scheme. If $r\geq k$, then a $k$-th order local truncation error is obtained at the irregular points. However, $r=k-1$ suffices to keep the global error to the $k$-th order \cite{GUSTAFSSON75}, and hence $r=3$ is required by the ADER 4 scheme when ${\bf S}={\bf 0}$. 

\item A special attention needs to be paid in the case of imperfect hydraulic contact (\ref{JCPerm}). Typical values of $\mathcal{K}$ range around $10^{-7}$ m/s/ Pa. From (\ref{JC0}) and (\ref{MatPerm}), it follows that numbers close to $10^7$ coexist with numbers close to 1 in the matrix ${\bf L}_1^0$. The matrices involved in steps 3 and 4 of the immersed interface method are then badly conditioned, which generates numerical instabilities during time-marching. To overcome this difficulty, we normalize the physical parameters and the unknowns in our codes. This normalization is described in appendix \ref{SecAnnexeNorme}. In practice, it is used whatever the interface conditions. 
\end{enumerate} 

%------------------------------------------------------------------------------------------
%------------------------------------------------------------------------------------------

\section{Numerical experiments}\label{SecRes}

\subsection{Physical parameters}

\begin{table}[htbp]
\begin{center}
\begin{tabular}{lll}
\hline
Saturating fluid & $\rho_f$ (kg/m$^3$)       & 1000                 \\
                 & $c$ (m/s)                 & 1500                 \\
                 & $\eta$ (Pa.s)             & 0 / $1.05\,10^{-3}$  \\
Grain            & $\rho_s$ (kg/m$^3$)       & 2690                 \\
                 & $\mu$ (Pa)                & $1.86\,10^9$         \\
%                 & $K_s$ (Pa)                &                     \\
Matrix           & $\phi$                    & 0.38                 \\ 
                 & $a$                       & 1.8                  \\
                 & $\kappa$ (m$^2$)          & $2.79\,10^{-11}$     \\
                 & $\lambda_0$ (Pa)          & $1.2\,10^8$          \\
                 & $m$ (Pa)                  & $5.34\,10^9$         \\
                 & $\beta$                   & 0.95                 \\
Interface        & $\mathcal{K}$ (m/s/Pa)       & $5.\,10^{-7}$        \\    
\hline
Phase velocities & $c_{pf}(f_0)$ (m/s)       & 2066.43              \\
                 & $c_{ps}(f_0)$ (m/s)       & 124.36               \\
				 & $c_s(f_0)$ (m/s)          & 953.05               \\
                 & $c_{pf}^\infty$ (m/s)     & 2071.85              \\
                 & $c_{ps}^\infty$ (m/s)     & 741.65               \\
                 & $c_s^\infty$ (m/s)        & 1006.32              \\ 
                 & $f_c$ (Hz)                & 1264.49              \\
\hline
\end{tabular}
\caption{Poroelastic medium $\Omega_1$: physical parameters and acoustic properties at $f_0=20$ Hz.}
\label{TabParametres}
\end{center}
\end{table}

The acoustic medium $\Omega_0$ is water ($\rho_f=1000$ kg/m$^3$, $c=1500$ m/s). The poroelastic medium $\Omega_1$ is a water saturated unconsolidated sand, whose material properties are summarized in table \ref{TabParametres} and are issued from table 3 of \cite{SIDLER10}. In some experiments, an inviscid saturating fluid is artificially considered: $\eta=0$ Pa.s, the other parameters being unchanged. It is mainly addressed here for a numerical purpose. The cases of open pores (\ref{JCOpen}), sealed pores (\ref{JCClose}), and imperfect pores (\ref{JCPerm}) are successively investigated. In the latter case, the value of hydraulic permeability $\mathcal{K}$ is given in table \ref{TabParametres}.

Once the spatial mesh sizes $\Delta\,x=\Delta\,y$ are chosen on the coarse grid, the time step follows from the CFL number in $\Omega_1$: $c_{pf}^\infty\,\Delta\,t\,/\,\Delta\,x=0.95<1$. The time evolution of the source is a combination of truncated sinusoids
\begin{equation}
h(t)=
\left\{
\begin{array}{l}
\displaystyle
\displaystyle \sum_{m=1}^4 a_m\,\sin(\beta_m\,\omega_0\,t)\quad \mbox{ if  }\, 0<t<\frac{\textstyle 1}{\textstyle f_0},\\
\\
0 \,\mbox{ otherwise}, 
\end{array}
\right.
\label{JKPS}
\end{equation}
where $\beta_m=2^{m-1}$, $\omega_0=2\pi\,f_0$; the coefficients $a_m$ are: $a_1=1$, $a_2=-21/32$, $a_3=63/768$, $a_4=-1/512$, ensuring $C^6$ smoothness of the solution. Two types of sources and boundary conditions are considered:
\begin{itemize}
\item no source term $f_p$ in (\ref{LCfluide}), but an incident plane wave in $\Omega_0$ as initial conditions:
\begin{equation}
{\bf U}(x,\,y,\,t_0)=-
\left(   
\begin{array}{ccc}
\displaystyle
\frac{\textstyle \cos \theta}{\textstyle c}\\
[8pt]
\displaystyle
\frac{\textstyle \sin \theta}{\textstyle c}\\
[8pt]
\displaystyle
\rho_f
\end{array}
\right)\,
h\left(t_0-\frac{\textstyle x\,\cos\theta+y\,\sin\theta}{\textstyle c}\right),
\label{OndePlane}
\end{equation}
where $\theta$ is the angle between the wavevector and the $x$-axis, and $t_0$ adjusts the location of the plane wave in $\Omega_0$. In section \ref{SecResT1}, the diffracted plane waves are computed exactly, and they are enforced numerically on the edges of the computational domain. In section \ref{SecResT2}, periodic computational edges are imposed along $y$-direction, and hence the incident acoustic wave does not need to be enforced;
\item null initial conditions, but a varying source term in (\ref{LCfluide})
\begin{equation}
f_p=\delta(x-x_s)\,\delta(y-y_s)\,h(t)
\label{PtSource}
\end{equation} 
that generates cylindrical waves. The size of the domain and the duration of the simulations are defined so that no special attention is required to simulate outgoing waves, for instance with Perfectly-Matched Layers \cite{MARTIN08}. 
\end{itemize}

%------------------------------------------------------------------------------------------

\subsection{Test 1: plane wave on a plane interface}\label{SecResT1}

\begin{figure}[htbp]
\begin{center}
\begin{tabular}{cc}
(a) & (b)\\
\includegraphics[scale=0.36]{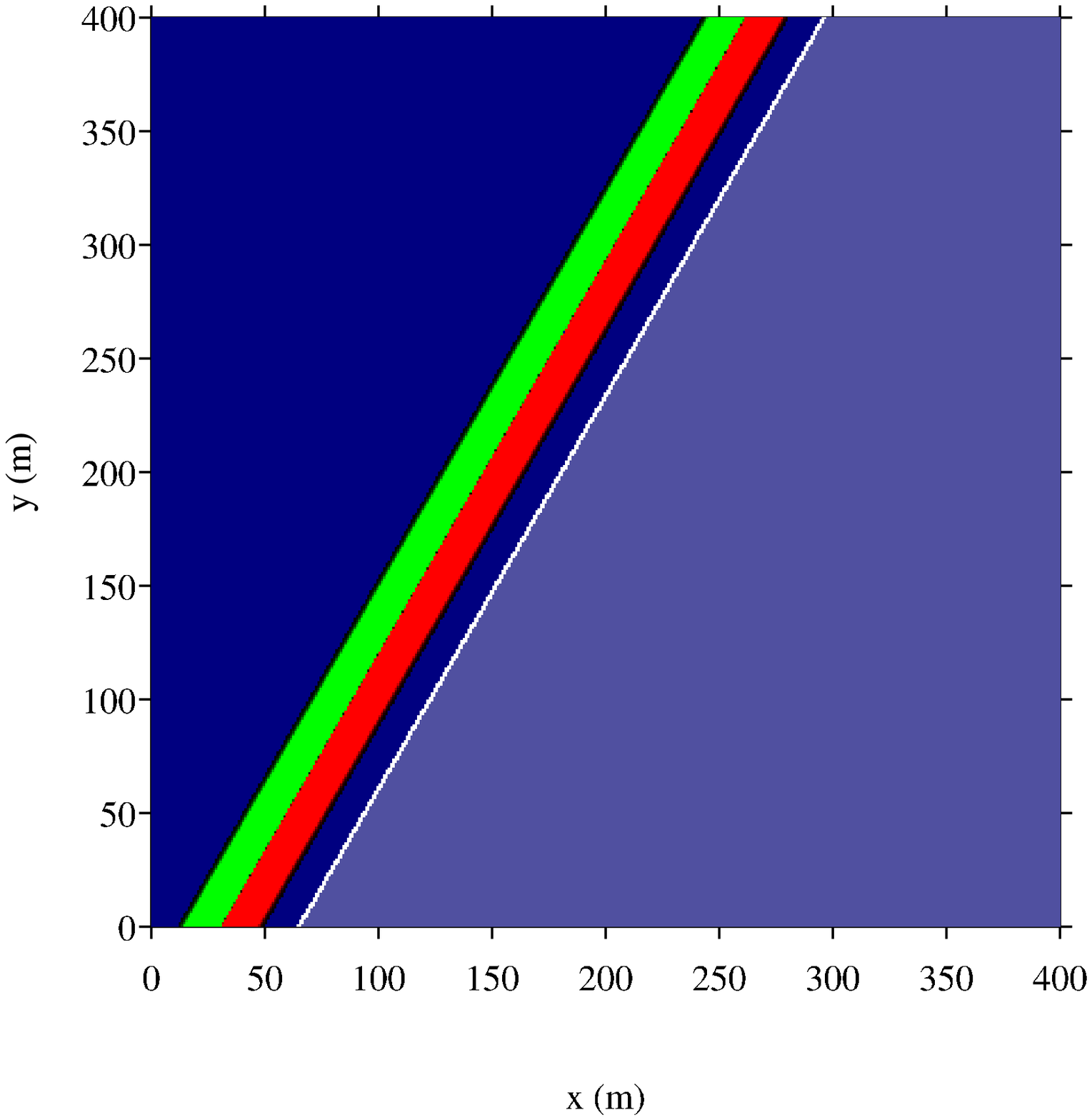}&
\includegraphics[scale=0.36]{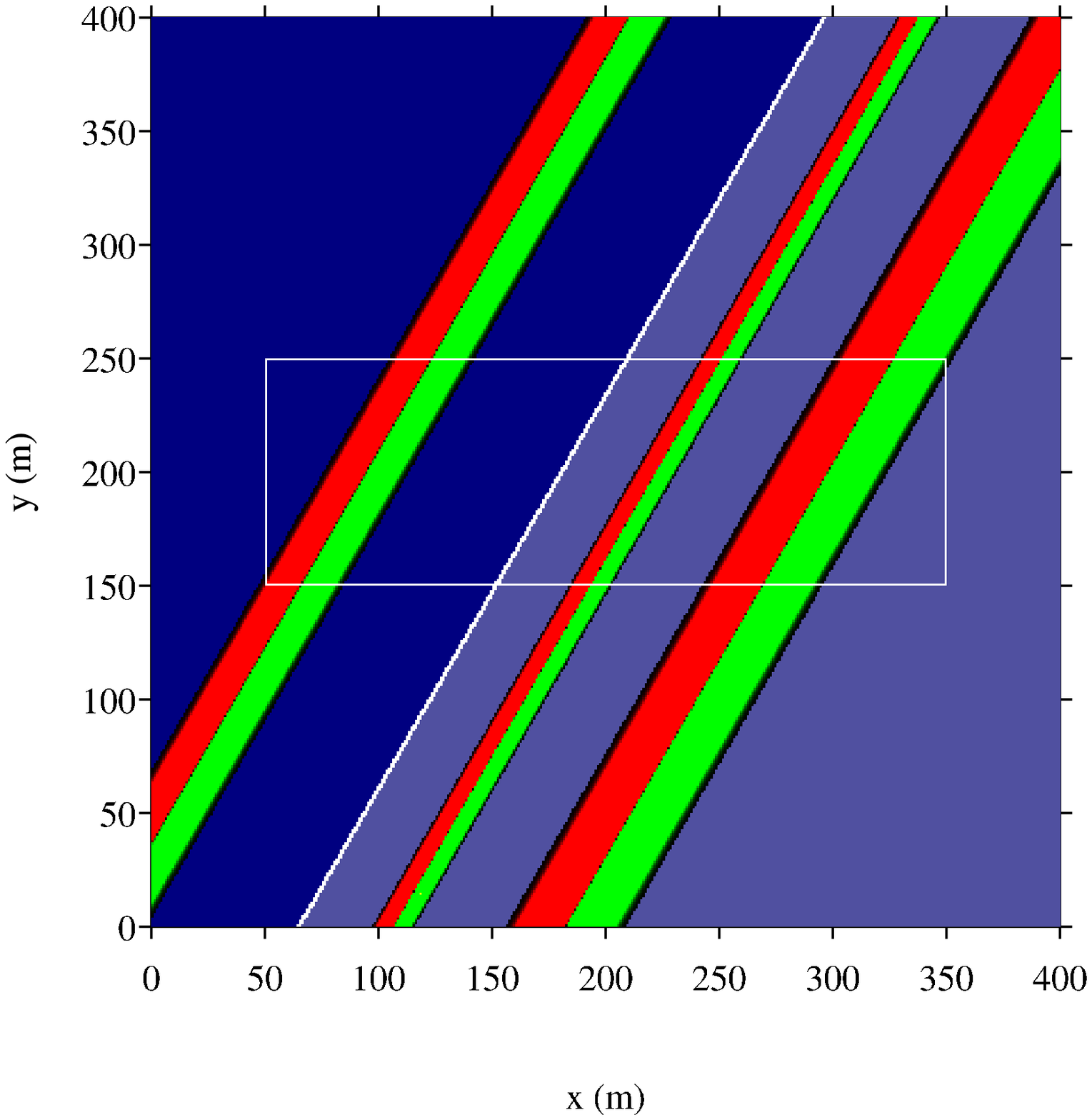}
\end{tabular}
\end{center}
\vspace{-0.8cm}
\caption{test 1. Snapshots of $p$ at the initial instant (a) and at the instant of measure, with open pore conditions (b). The white rectangle denotes the zone where convergence errors are measured.}
\label{FigTest1Cartes}
\end{figure}

As a first test, we consider a domain $[0,\,400]$ m$^2$ cut by a plane interface $\Gamma$ with slope 60 degrees. An incident plane wave propagates in the fluid, with $t_0=0.03$ s and $\theta=-30$ degrees (\ref{OndePlane}). Consequently, the incident wave crosses the interface normally, leading to a 1-D configuration (figure \ref{FigTest1Cartes}-a); from a numerical point of view, however, the problem is fully bidimensional. The advantage of such a 1-D configuration is that each diffracted wave has interacted with the interface and is consequently very sensitive to the discretization of the interface conditions (\ref{JC}). The central frequency $f_0=40$ Hz is much smaller than the critical Biot frequency: see (\ref{Fc}) and table \ref{TabParametres}.    in (\ref{JKPS}) is 

\begin{figure}[htbp]
\begin{center}
\begin{tabular}{ccc}
(a) & (b)\\
\includegraphics[scale=0.34]{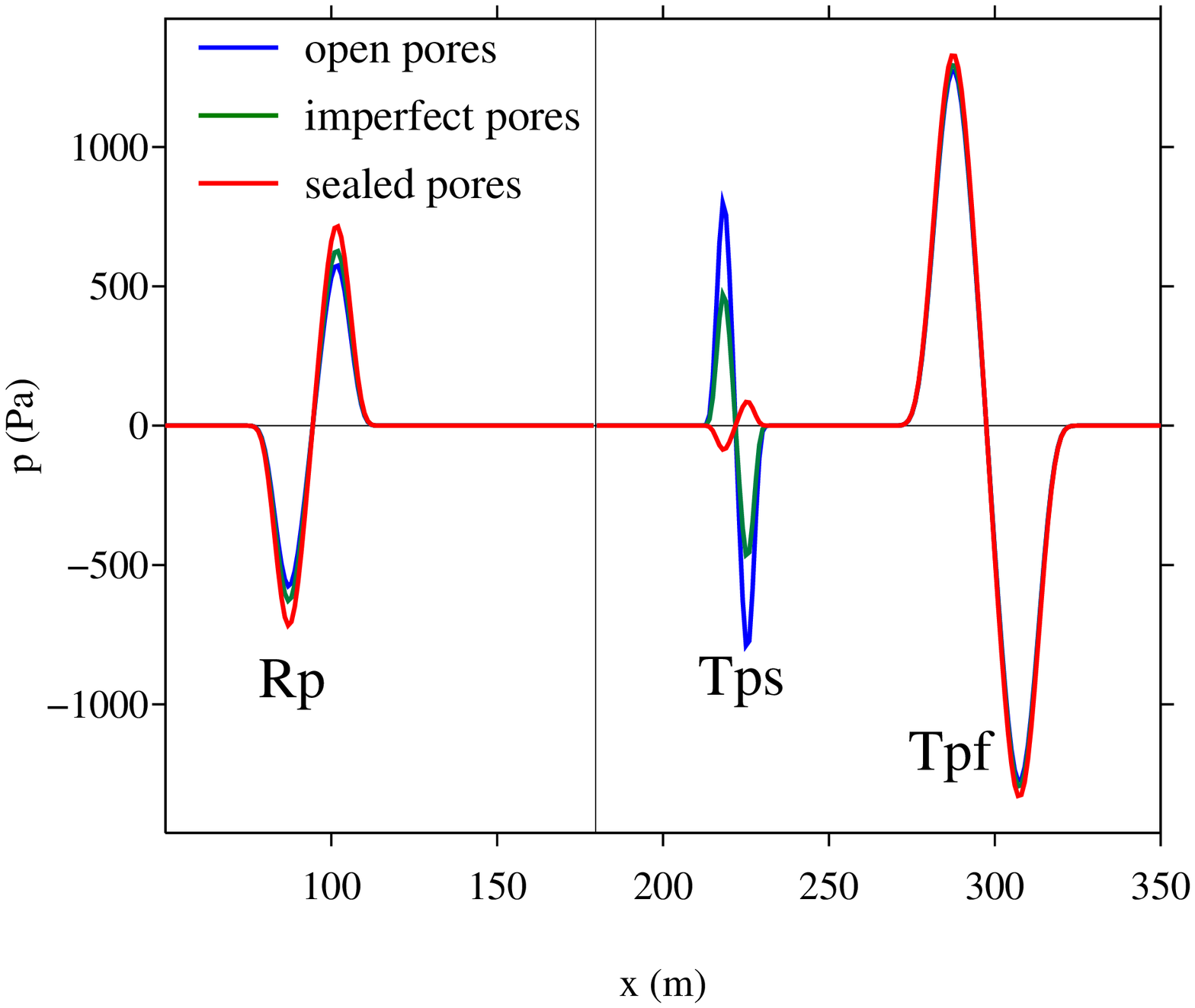}&
\includegraphics[scale=0.34]{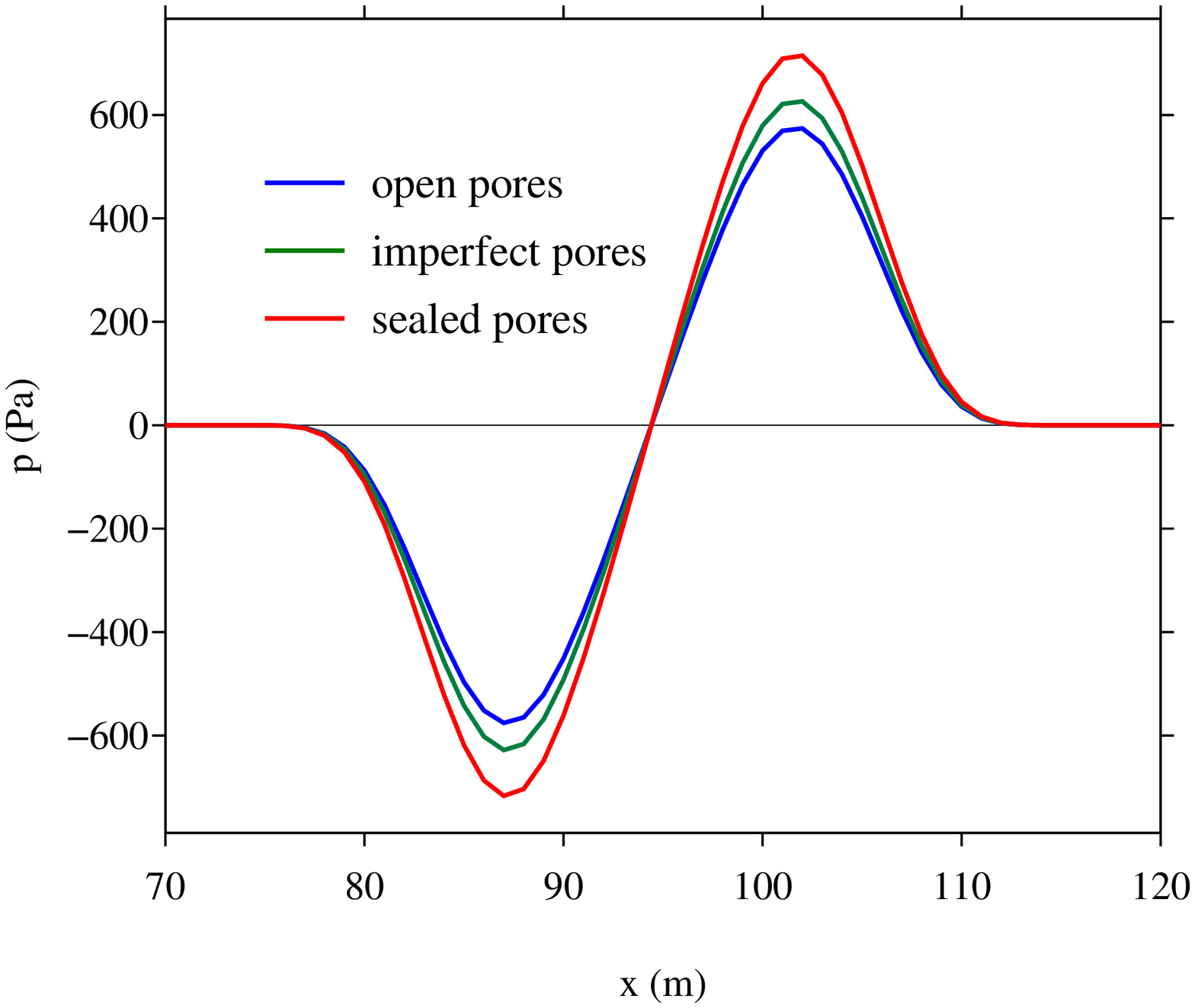}
\end{tabular}
\end{center}
\vspace{-0.8cm}
\caption{test 1. Slices of $p$ at the final instant, for various interface conditions. (a): reflected acoustic wave $R_p$, transmitted fast compressional wave $T_{pf}$, transmitted slow compressional wave $T_{ps}$. (b): zoom around the wave reflected in the fluid domain ($Rp$).}
\label{FigTest1Compare}
\end{figure}

We consider an inviscid saturating fluid ($\eta=0$ Pa.s): as a consequence, exact solutions are computed very accurately without Fourier synthesis, and splitting errors of the scheme are avoided (section \ref{SecNumLC}). The computations are done on a uniform grid of $400^2$ points, during 150 time steps. Comparisons with the exact values of the pressure $p$ are done on the sub-domain $[50,\,350]\mbox{ m} \times[150,\,250]$ m, in order to avoid spurious effects induced by the edges of the computational domain (figure \ref{FigTest1Cartes}-b). One observes the reflected acoustic wave, the transmitted fast wave and the transmitted slow compressional waves (no shear wave is generated in 1D). 

\begin{table}[htbp]
\begin{center}
\begin{small}
\begin{tabular}{l|l|ll|ll|ll}
interface condition & $N$    & $r=1$  & order & $r=2$ & order & $r=3$ & order    \\
\hline
open pores 
& 400  & $4.987\,10^{0}$ & {\bf -}    & $2.855\,10^{0}$  & {\bf -}     & $2.008\,10^{0}$  & {\bf -}\\
& 800  & $1.154\,10^{0}$ & {\bf 2.212}& $4.042\,10^{-1}$ & {\bf 2.820} & $2.207\,10^{-1}$ & {\bf 3.186}\\
& 1200 & $4.997\,10^{-1}$& {\bf 2.064}& $1.160\,10^{-1}$ & {\bf 3.079} & $5.052\,10^{-2}$ & {\bf 3.636}\\
& 1600 & $2.791\,10^{-1}$& {\bf 2.025}& $4.728\,10^{-2}$ & {\bf 3.120} & $1.661\,10^{-2}$ & {\bf 3.867}\\
& 2000 & $1.788\,10^{-1}$& {\bf 1.996}& $2.367\,10^{-2}$ & {\bf 3.101} & $6.870\,10^{-3}$ & {\bf 3.956}\\
& 2400 & $1.239\,10^{-1}$& {\bf 2.012}& $1.345\,10^{-2}$ & {\bf 3.100} & $3.318\,10^{-3}$ & {\bf 3.992}\\
\hline
sealed pores 
& 400  & $5.226\,10^{0}$ & {\bf -}    & $7.214\,10^{-1}$ & {\bf -}     & $6.147\,10^{-1}$ & {\bf -}  \\
& 800  & $1.467\,10^{0}$ & {\bf 1.833}& $8.656\,10^{-2}$ & {\bf 3.059} & $4.886\,10^{-2}$ & {\bf 3.653}\\
& 1200 & $6.704\,10^{-1}$& {\bf 1.931}& $2.449\,10^{-2}$ & {\bf 3.114} & $1.031\,10^{-2}$ & {\bf 3.837}\\
& 1600 & $3.809\,10^{-1}$& {\bf 1.965}& $1.002\,10^{-2}$ & {\bf 3.106} & $3.346\,10^{-3}$ & {\bf 3.912}\\
& 2000 & $2.449\,10^{-1}$& {\bf 1.979}& $5.029\,10^{-3}$ & {\bf 3.089} & $1.389\,10^{-3}$ & {\bf 3.940}\\
& 2400 & $1.706\,10^{-1}$& {\bf 1.983}& $2.868\,10^{-3}$ & {\bf 3.080} & $6.764\,10^{-4}$ & {\bf 3.947}\\
\hline
imperfect pores 
& 400  & $4.826\,10^{0}$ & {\bf -}    & $1.739\,10^{0}$  & {\bf -}     & $1.262\,10^{}$  & {\bf -} \\
& 800  & $1.200\,10^{0}$ & {\bf 2.008}& $2.412\,10^{-1}$ & {\bf 2.850} & $1.287\,10^{-1}$ & {\bf 3.294}\\
& 1200 & $5.233\,10^{-1}$& {\bf 2.047}& $6.882\,10^{-2}$ & {\bf 3.093} & $2.913\,10^{-2}$ & {\bf 3.664}\\
& 1600 & $2.916\,10^{-1}$& {\bf 2.033}& $2.792\,10^{-2}$ & {\bf 3.136} & $9.517\,10^{-3}$ & {\bf 3.889}\\
& 2000 & $1.858\,10^{-1}$& {\bf 2.020}& $1.389\,10^{-2}$ & {\bf 3.129} & $3.922\,10^{-3}$ & {\bf 3.973}\\
& 2400 & $1.285\,10^{-1}$& {\bf 2.022}& $7.891\,10^{-3}$ & {\bf 3.101} & $1.891\,10^{-3}$ & {\bf 4.001}
\end{tabular}
\end{small} 
\end{center}
\vspace{-0.5cm}
\caption{test 1. Error measurements and convergence rate in $l_2$ norm, for various interface conditions. Linear $(r=1)$, quadratic $(r=2)$ or cubic $(r=3)$ immersed interface method. Measures are done after $3\,N\,/\,8$ time steps.}
\label{TabTest1Convergence}
\end{table}

The influence of the interface conditions on the diffracted waves is illustrated in figure \ref{FigTest1Compare}. The exact value of $p$ at the final instant is shown at $y=200$ m for the conditions (\ref{JCOpen}), (\ref{JCClose}) and (\ref{JCPerm}). The main difference between these three cases is observed in the transmitted slow wave. In real experiments, however, only the reflected acoustic wave is measured: a zoom on this wave is given in figure \ref{FigTest1Compare}-b. The hydraulic permeability $\mathcal{K}=5.\,10^{-7}$ leads to an intermediate regime between open and sealed pores. If $\mathcal{K}\geq 10^{-5}$, the results (not shown here) cannot be distinguished from those obtained with open pores. In the same way, results obtained with $\mathcal{K}\leq 10^{-8}$ cannot be distinguished from the sealed pores. 

\begin{figure}[htbp]
\begin{center}
\begin{tabular}{cc}
open pores (a) & open pores (b)\\
\includegraphics[scale=0.34]{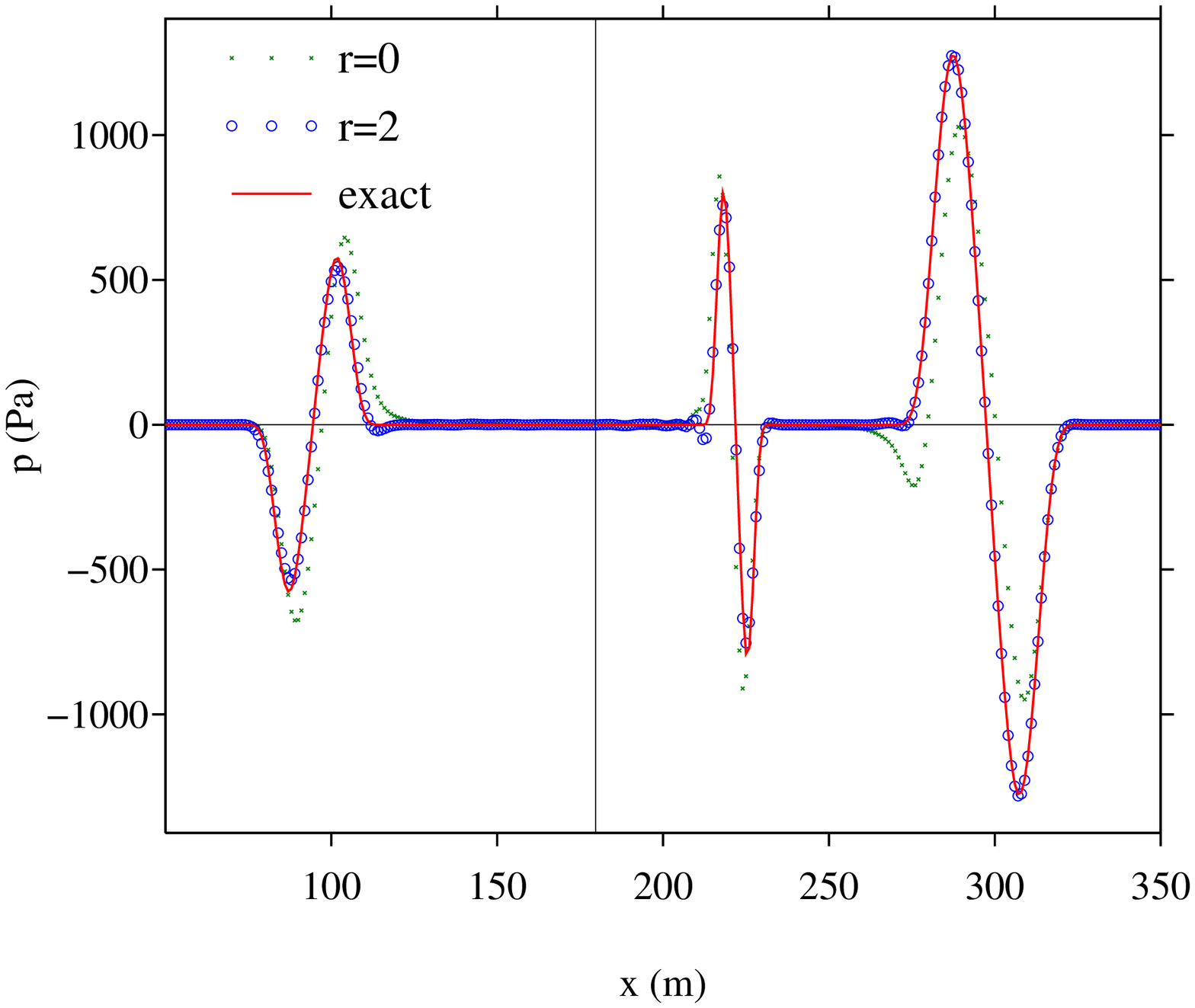}&
\includegraphics[scale=0.34]{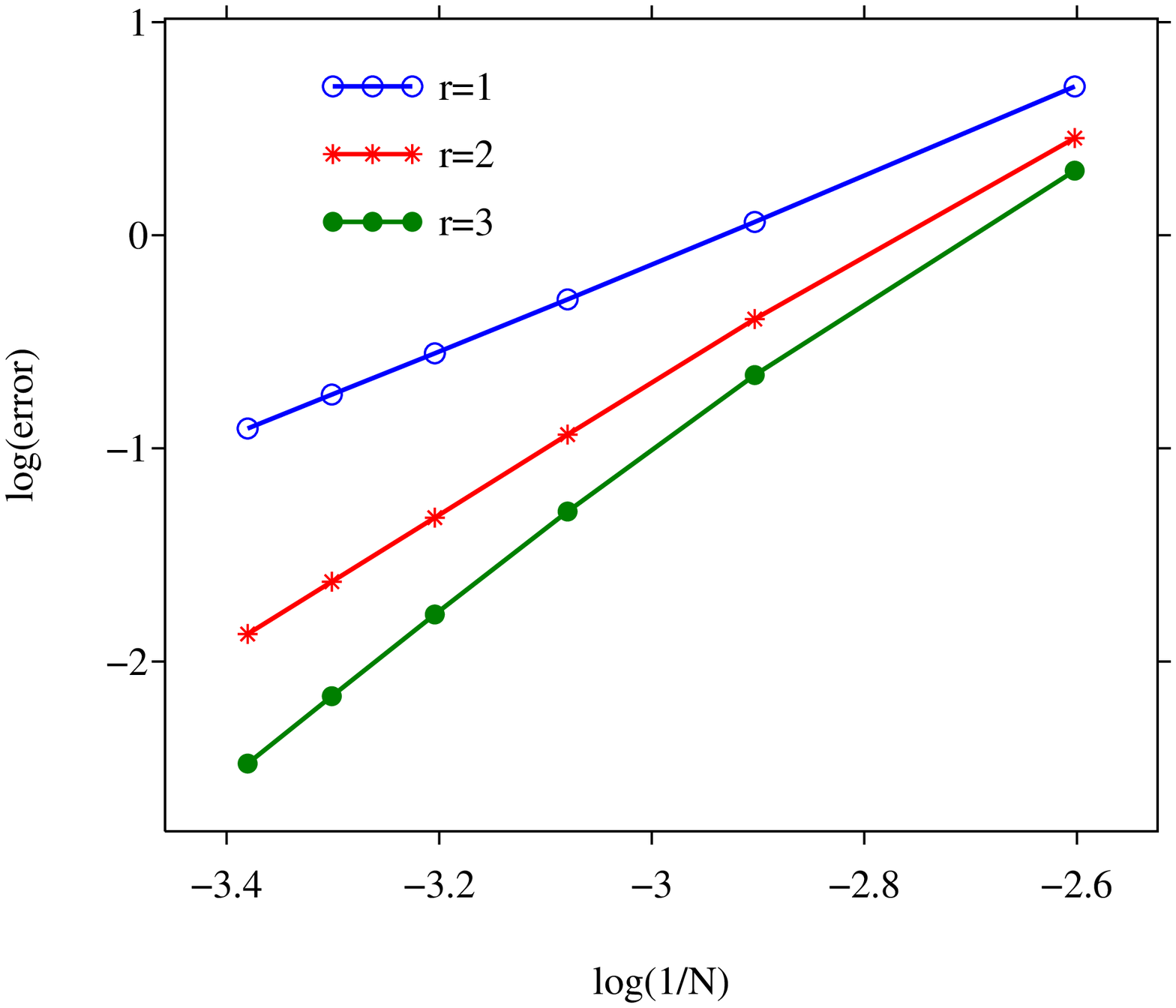}\\
sealed pores (c) & sealed pores (d)\\
\includegraphics[scale=0.34]{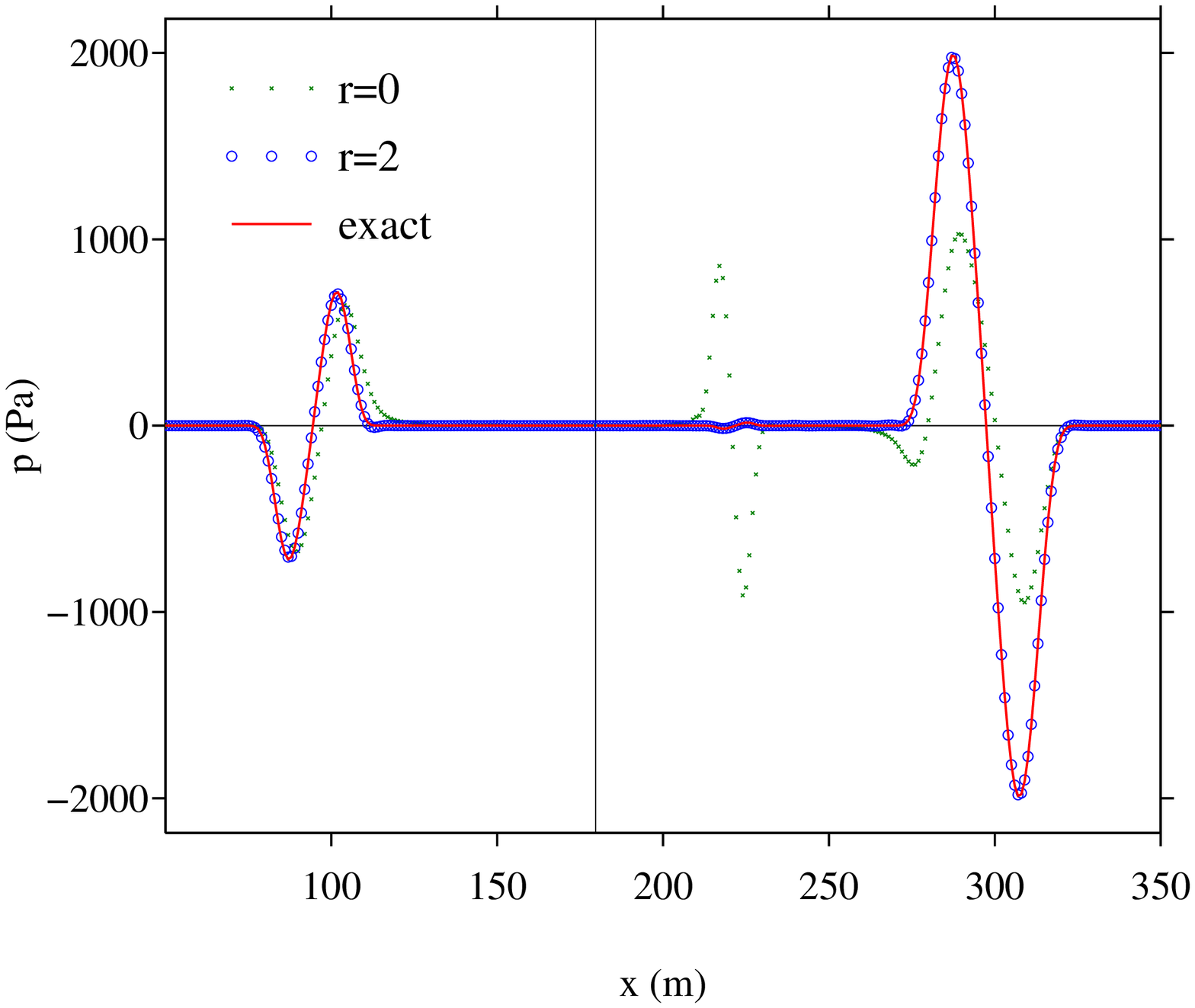}&
\includegraphics[scale=0.34]{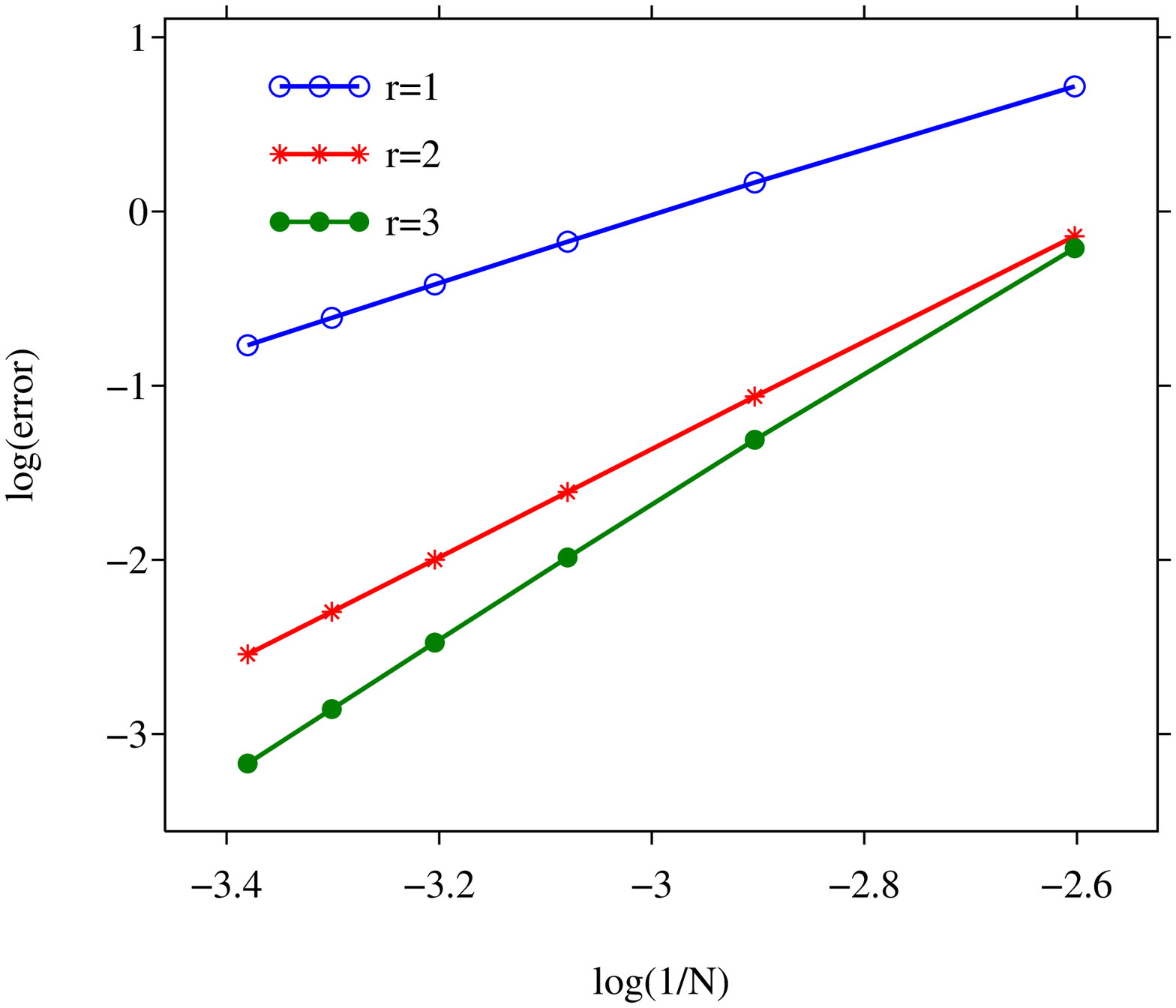}\\
imperfect pores (e) & imperfect pores (f)\\
\includegraphics[scale=0.34]{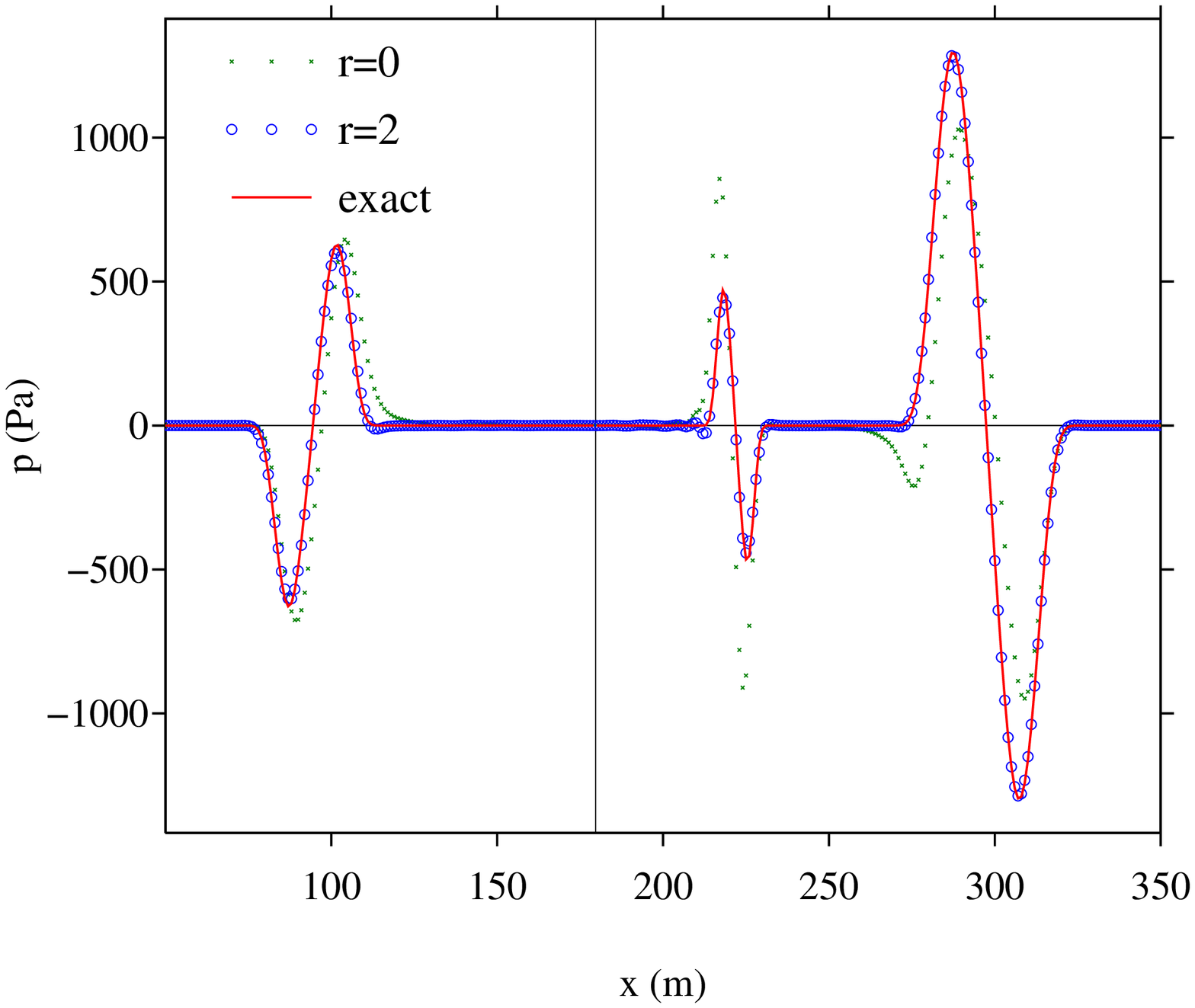}&
\includegraphics[scale=0.34]{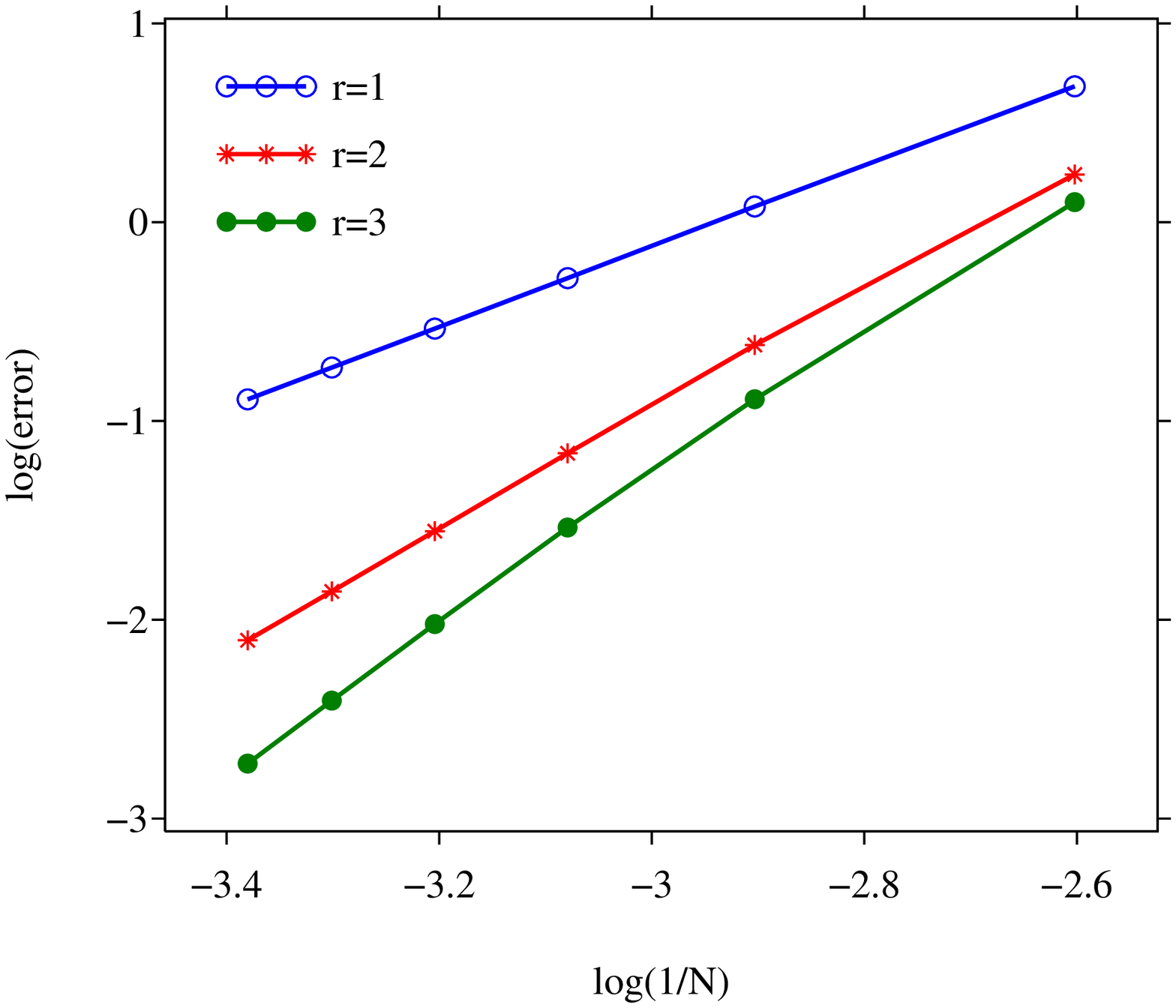}
\end{tabular}
\end{center}
\vspace{-0.8cm}
\caption{test 1. Left row: comparisons between exact values and numerical values of $p$. Right row: errors measured in $l_2$ norm versus the number of grid nodes, with various order $r$ of the immersed interface method.}
\label{FigTest1Ordre}
\end{figure}

How accurate is the discretization of the interface conditions is assessed through comparisons with the analytical solutions (figure \ref{FigTest1Ordre}). In the left row, such comparisons are proposed in two cases; $r=0$ means that no numerical treatment is done along the interface, and the numerical solution does not converge towards the exact one. On the contrary, excellent agreement is observed when $r=2$ is used (section \ref{SecNumIIM}).

Error measurements on successive refined grids are given in table \ref{TabTest1Convergence}. Convergence rates are drawn on the right row of figure \ref{FigTest1Ordre}. Various values of the order $r$ of the immersed interface method are investigated. As stated in section \ref{SecNumIIM}, fourth-order accuracy is maintained if third-order extrapolations ($r=3$) are used in the immersed interface method. From now on, all the simulations are performed with $r=3$.

\begin{figure}[htbp]
\begin{center}
\begin{tabular}{cc}
(a) & (b)\\
\includegraphics[scale=0.34]{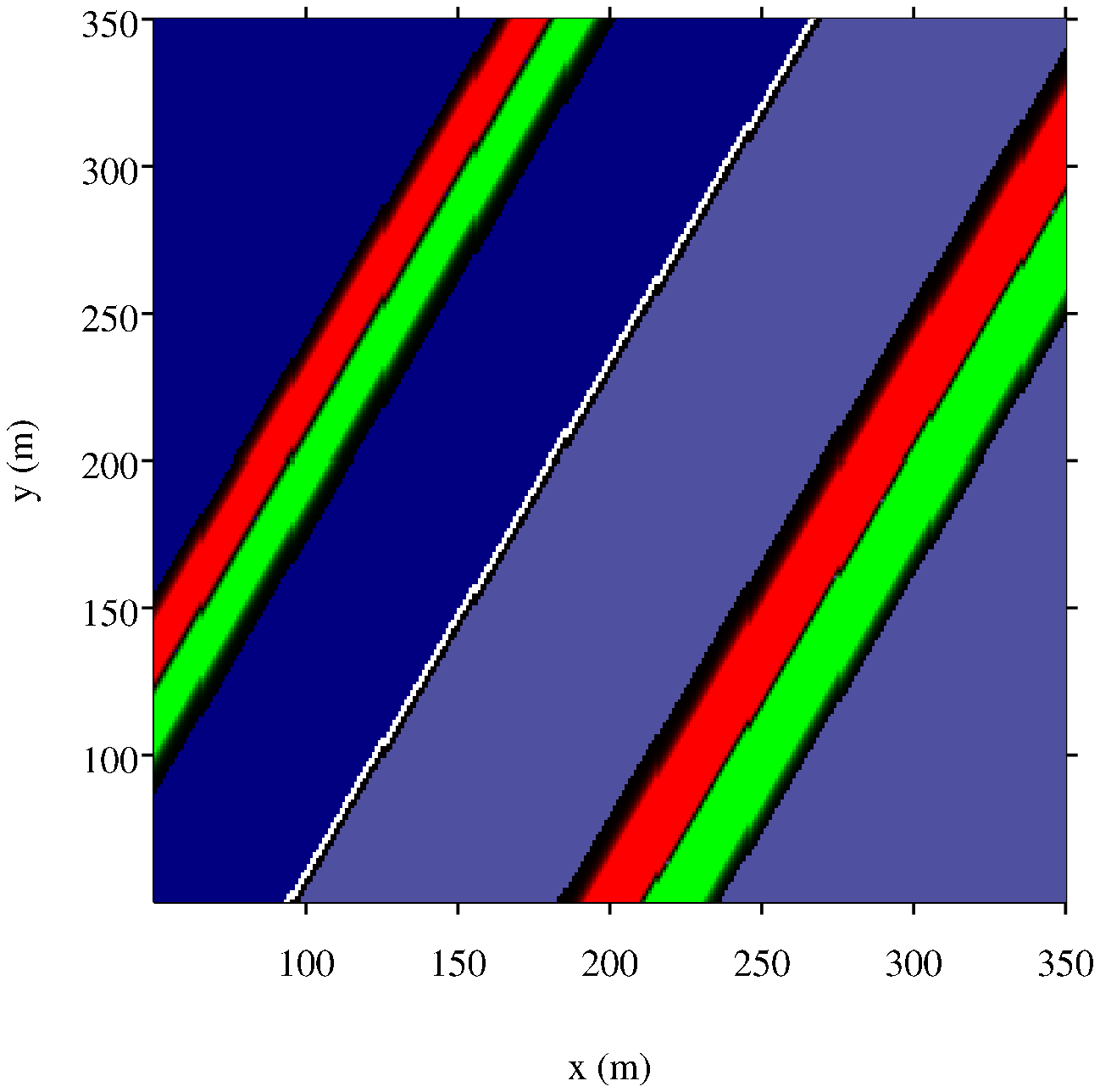}&
\includegraphics[scale=0.34]{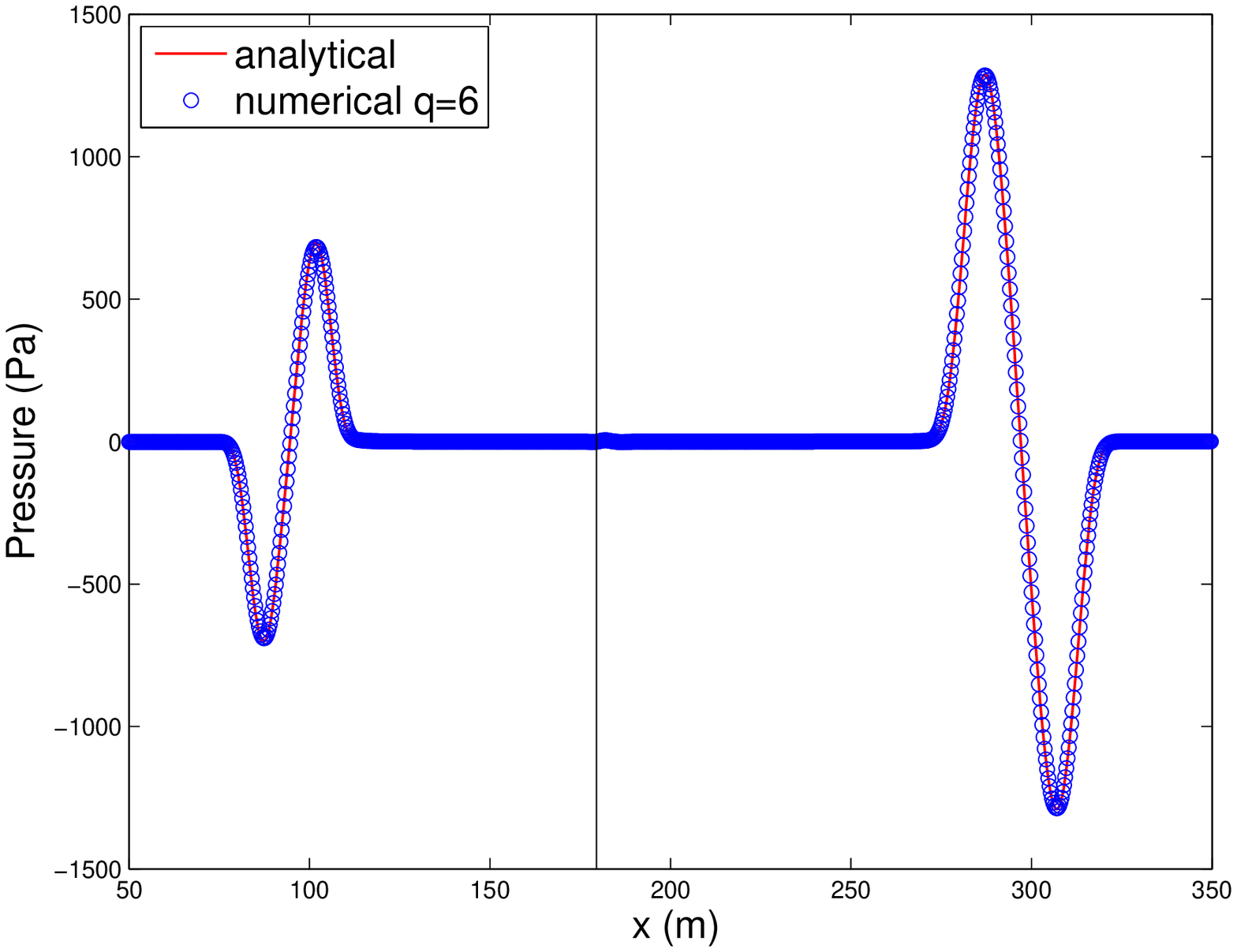}\\
(c) & (d)\\
\includegraphics[scale=0.34]{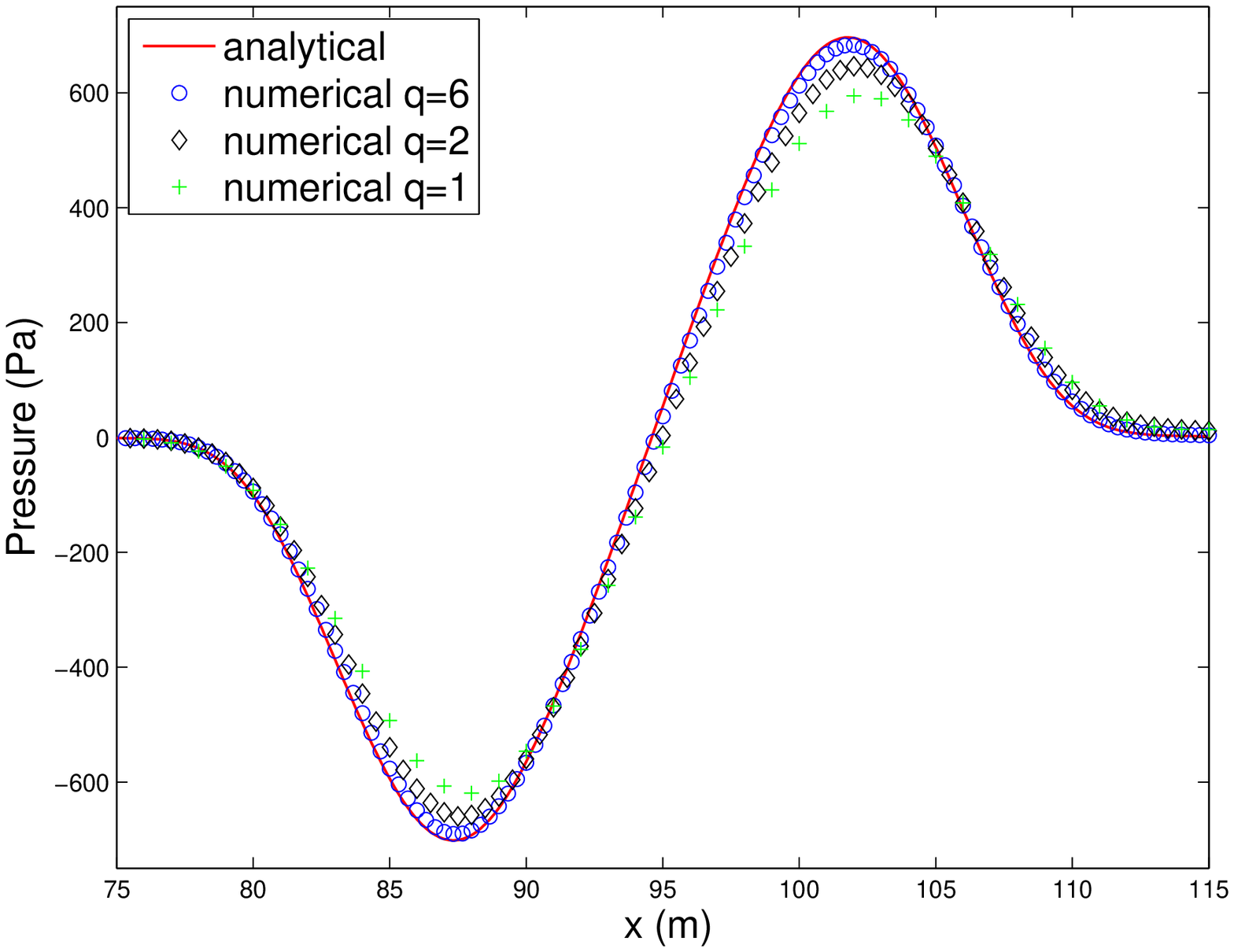}&
\includegraphics[scale=0.34]{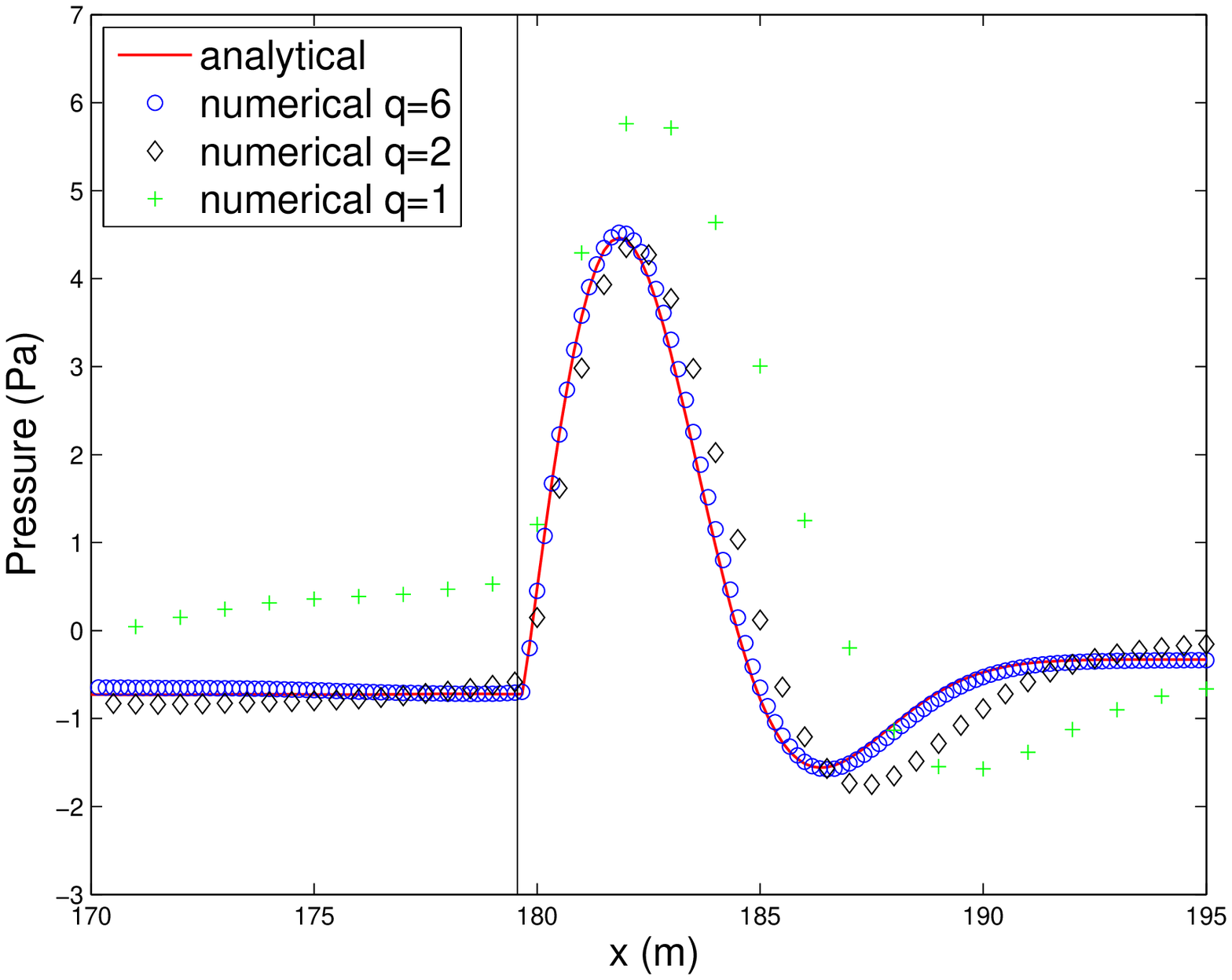}
\end{tabular}
\end{center}
\vspace{-0.8cm}
\caption{test 1, with a viscous saturating fluid ($\eta=1.05\,10^{-3}$ Pa.s) and open pore conditions. Snapshot (a) and slice (b) of the exact value of $p$ at the final instant; zooms on the reflected acoustic wave (c) and on the transmitted slow compressional wave (d), for various grid sizes. Vertical  line represents the location of the fluid / porous interface.}
\label{FigTest1Visqueux}
\end{figure}

In figure \ref{FigTest1Visqueux}, we take into account the viscosity of the the saturating fluid ($\eta=1.05\,10^{-3}$ Pa.s), the other parameters being unchanged (open pore conditions are considered). In this case, the slow compressional wave is a static mode that remains where it is generated, close to the interface in the porous medium, and hence a much finer grid is required (section \ref{SecNumAMR}). The exact solution is computed by Fourier synthesis on 256 modes, with a frequency step $\Delta\,f=1.6$ Hz.

In comparison with figures \ref{FigTest1Cartes}-(b) and \ref{FigTest1Ordre}-(a), the transmitted slow compressional wave cannot be seen in figure \ref{FigTest1Visqueux}-(b). Zooms around the reflected acoustic wave (c) and around the transmitted slow compressional wave (d) are proposed for various grids. If the coarse grid of the inviscid case is used ($q=1$), then large errors are observed compared with the exact solution. It is a consequence of the too crude discretization of small-scale slow wave (d). On the contrary, if a sufficiently fine grid is used ($q=6$),  an excellent agreement is obtain between  numerical and analytical solutions.

%------------------------------------------------------------------------------------------
%------------------------------------------------------------------------------------------

\subsection{Test 2: plane wave on a circular interface}\label{SecResT2}

A circular interface $\Gamma$ of radius $100$ m is centered at $(0,\,0)$ in the domain $[-600, 600]\,\rm{m}^2$ (figure \ref{FigTest2Init}). This configuration is relevant to examine the discretization of an interface with constant non-zero curvature. The cylinder is filled with the porous medium $\Omega_1$, while the acoustic medium $\Omega_0$ lies outside. The computational domain is discretized on $800^2$ points, leading to $\Delta\,x=\Delta\,y=1.5$ m. A locally refined mesh $[-110, 110]\,\rm{m}^2$ is used around the interface in order to compute accurately the different wave conversions. As in test 1, the source is an acoustic plane wave initially in $\Omega_0$, with $\theta=0$ degree (\ref{OndePlane}). The central frequency in (\ref{JKPS}) is $f_0=20$ Hz. The initial conditions are illustrated in figure \ref{FigTest2Init}.

\begin{figure}[htbp]
\begin{center}
\begin{tabular}{cc}
(a) & (b)\\
\includegraphics[scale=0.4]{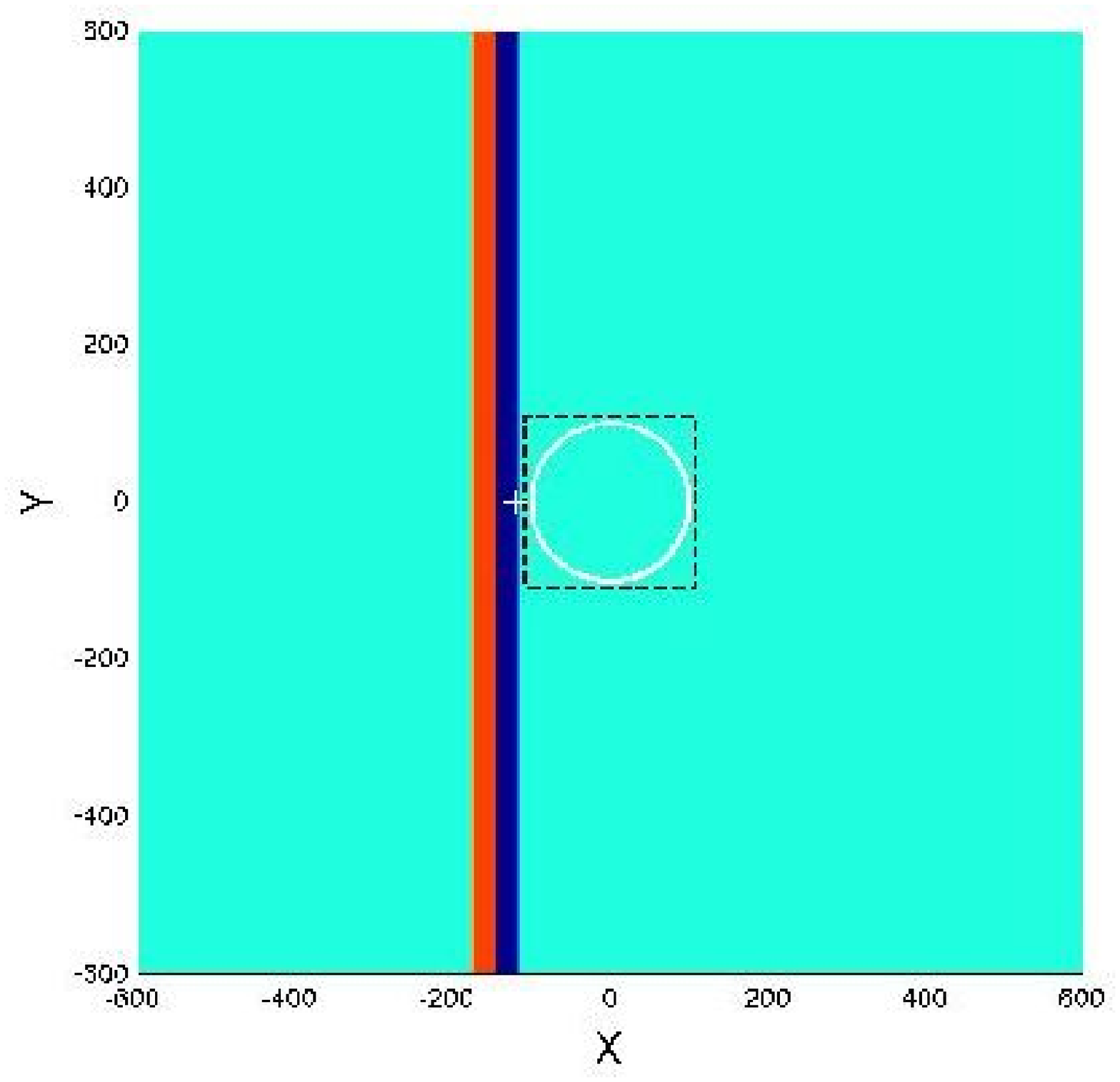}&
\includegraphics[scale=0.4]{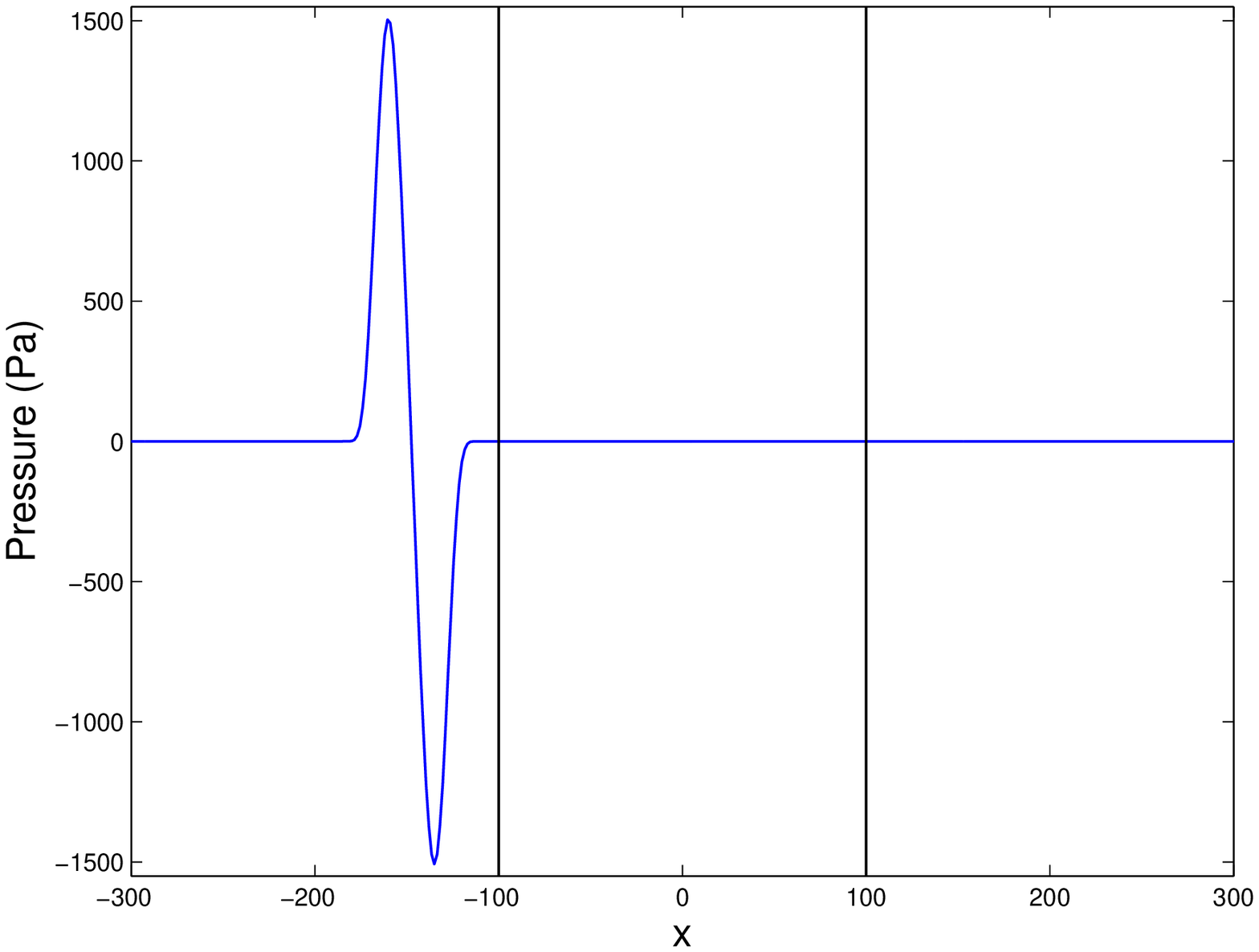}
\end{tabular}
\end{center}
\vspace{-0.8cm}
\caption{test 2. Initial values of $p$: snapshot (a) and slice along the $x$-axis, at $y=0$ m (b). The dotted square in (a) represents the frontiers of the refined mesh. The vertical lines in (b) denote the location of interfaces.}
\label{FigTest2Init}
\end{figure}

\begin{figure}[htbp]
\begin{center}
\begin{tabular}{cc}
open pores & open pores\\
\includegraphics[scale=0.4]{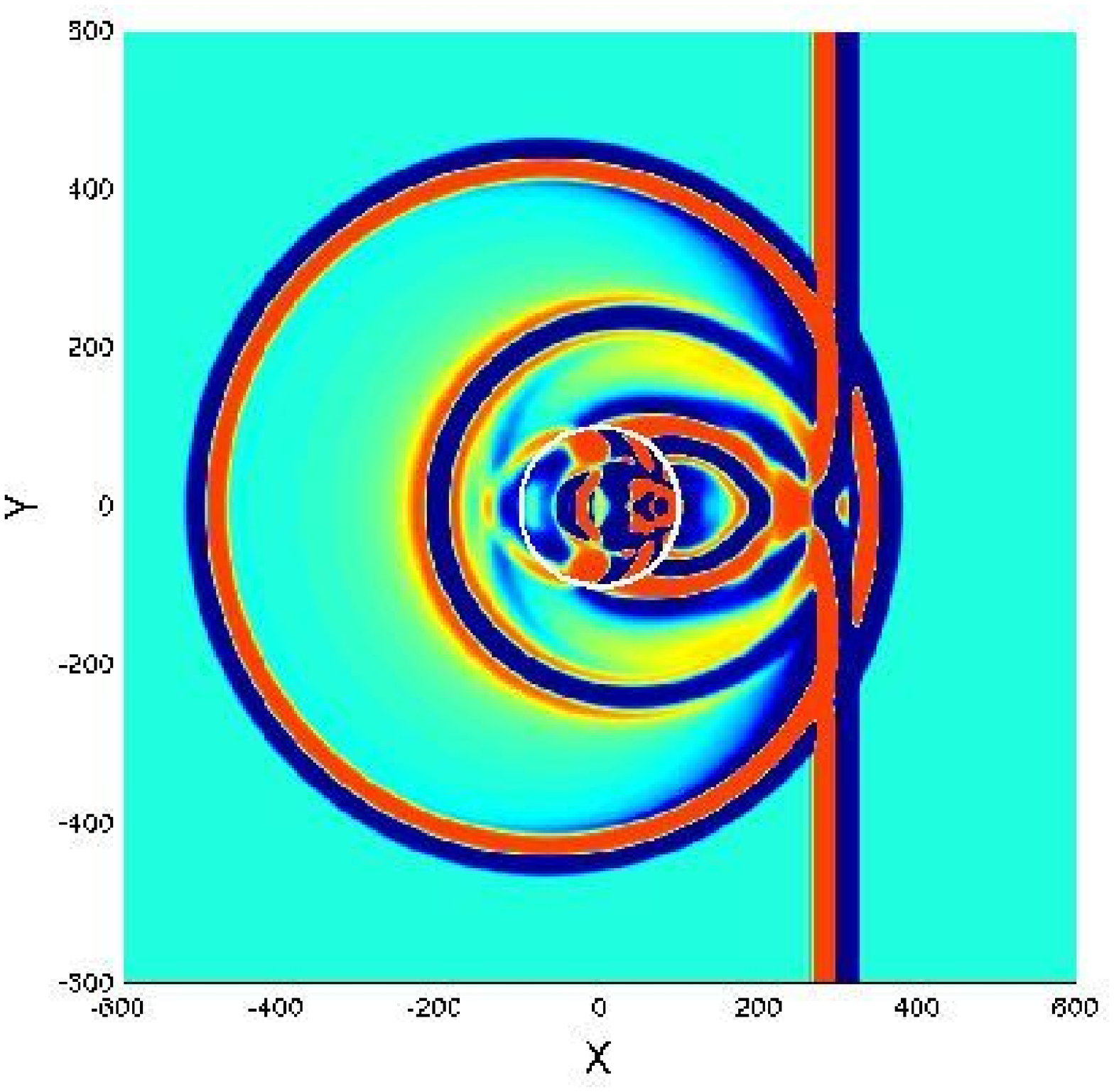}&
\includegraphics[scale=0.4]{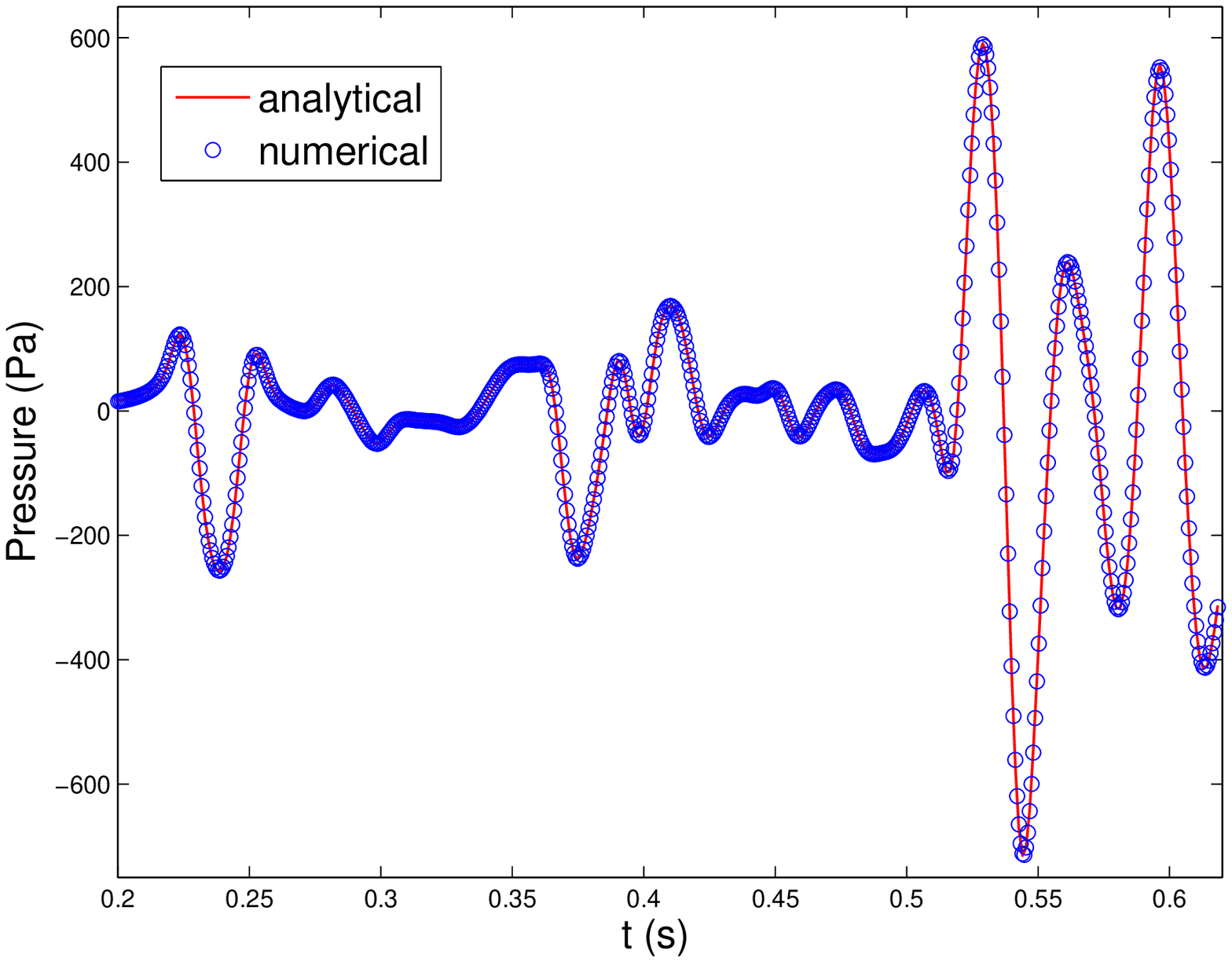}\\
sealed pores & sealed pores\\
\includegraphics[scale=0.4]{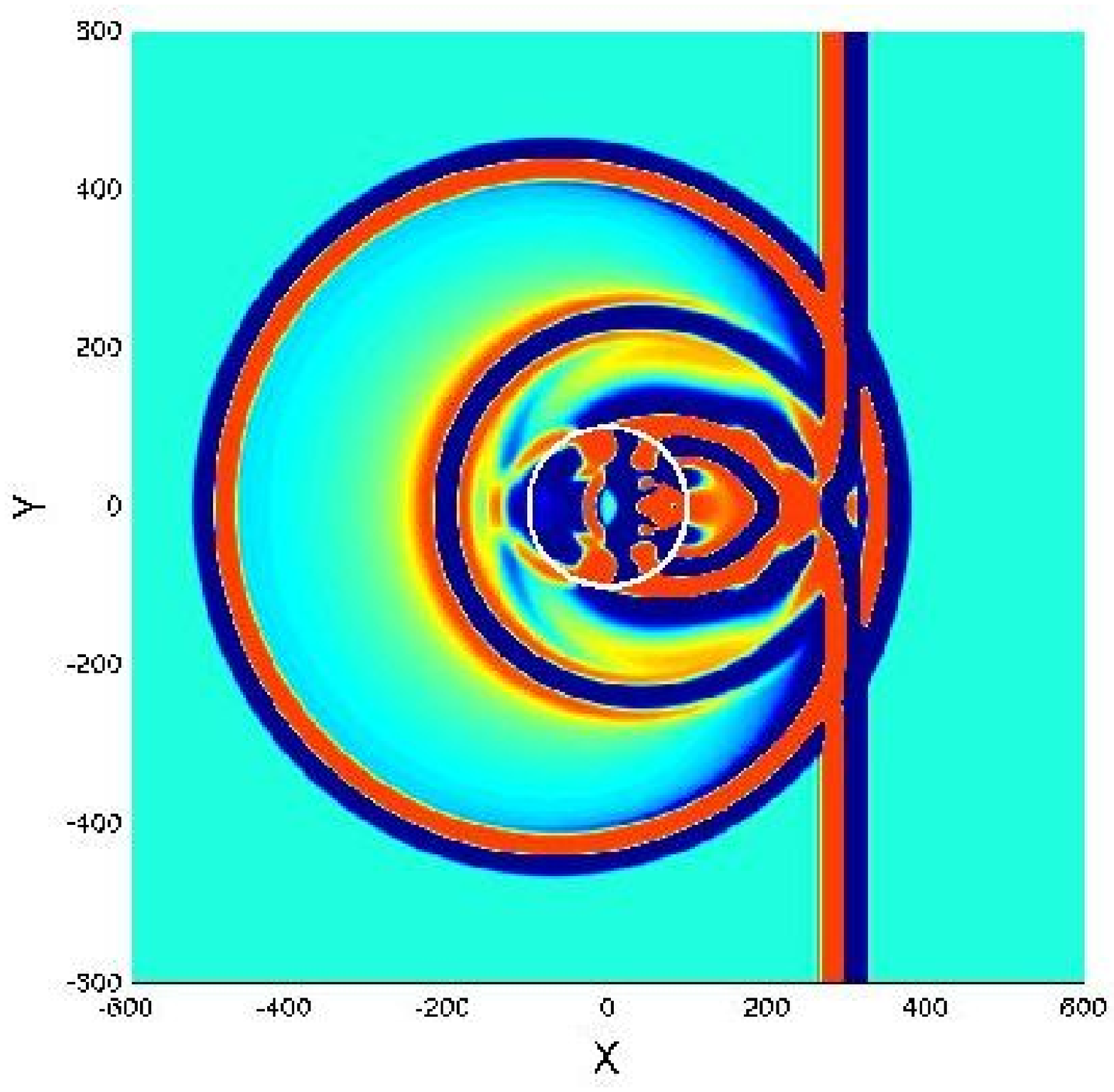}&
\includegraphics[scale=0.4]{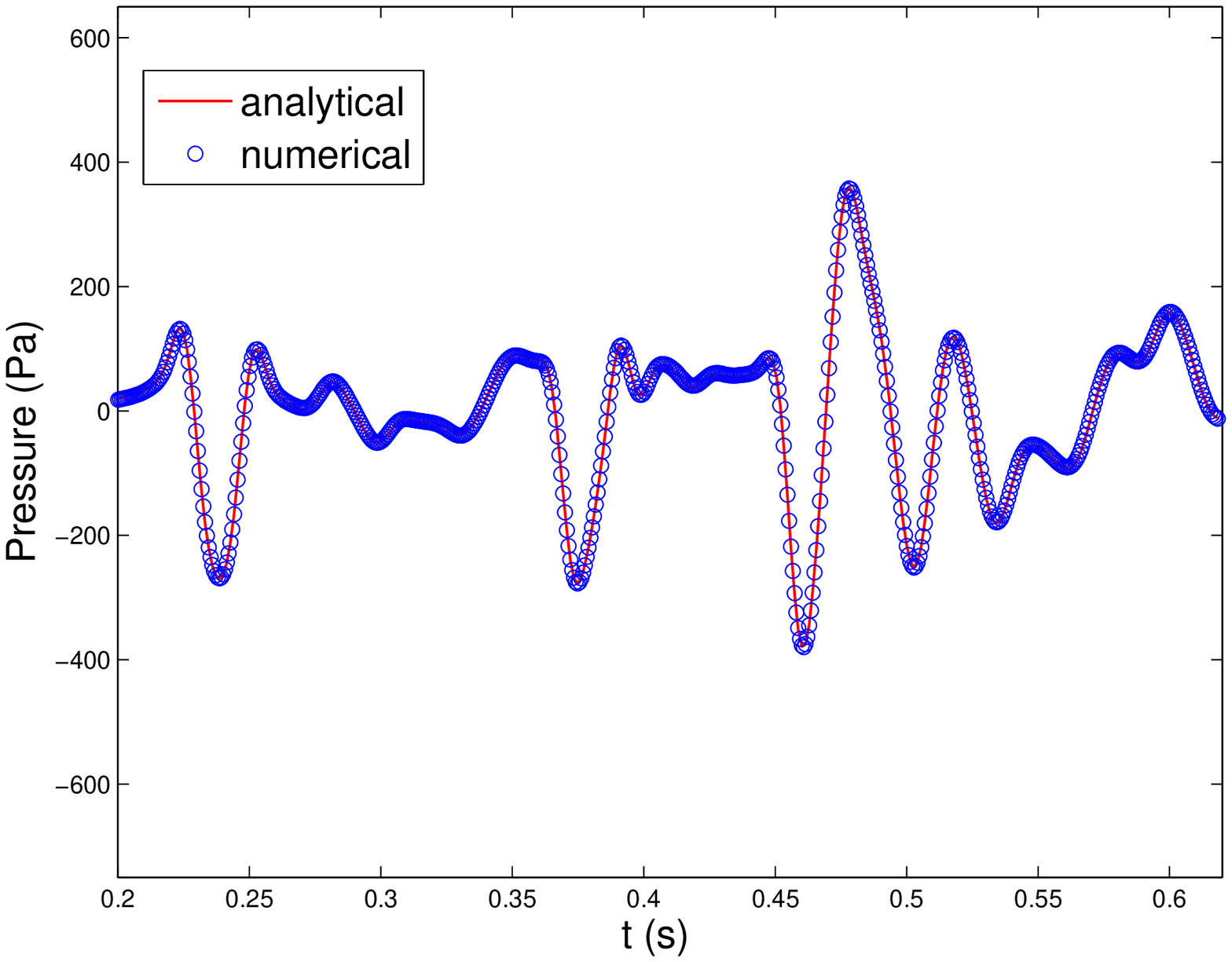}\\
imperfect pores & imperfect pores\\
\includegraphics[scale=0.4]{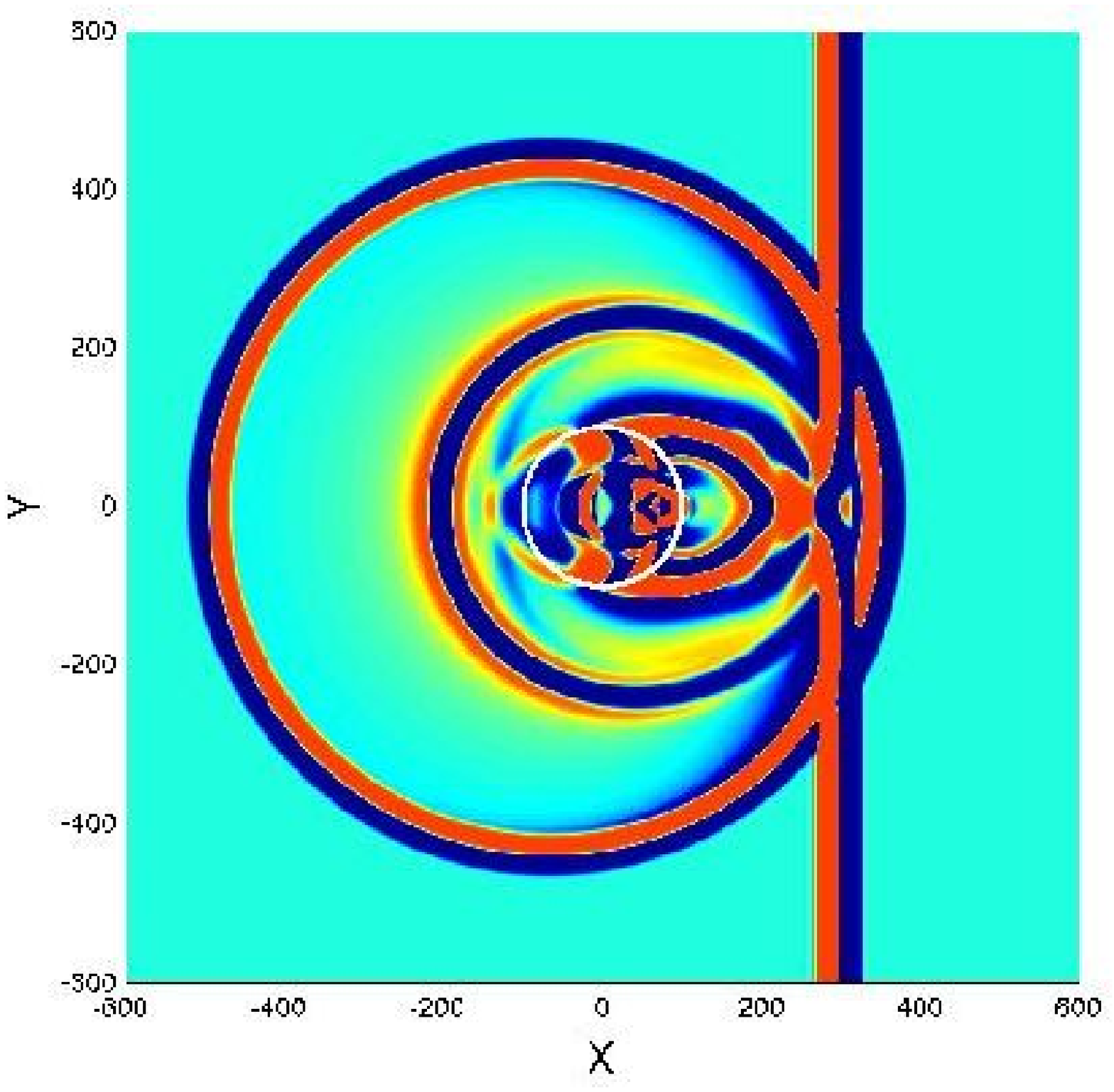}&
\includegraphics[scale=0.4]{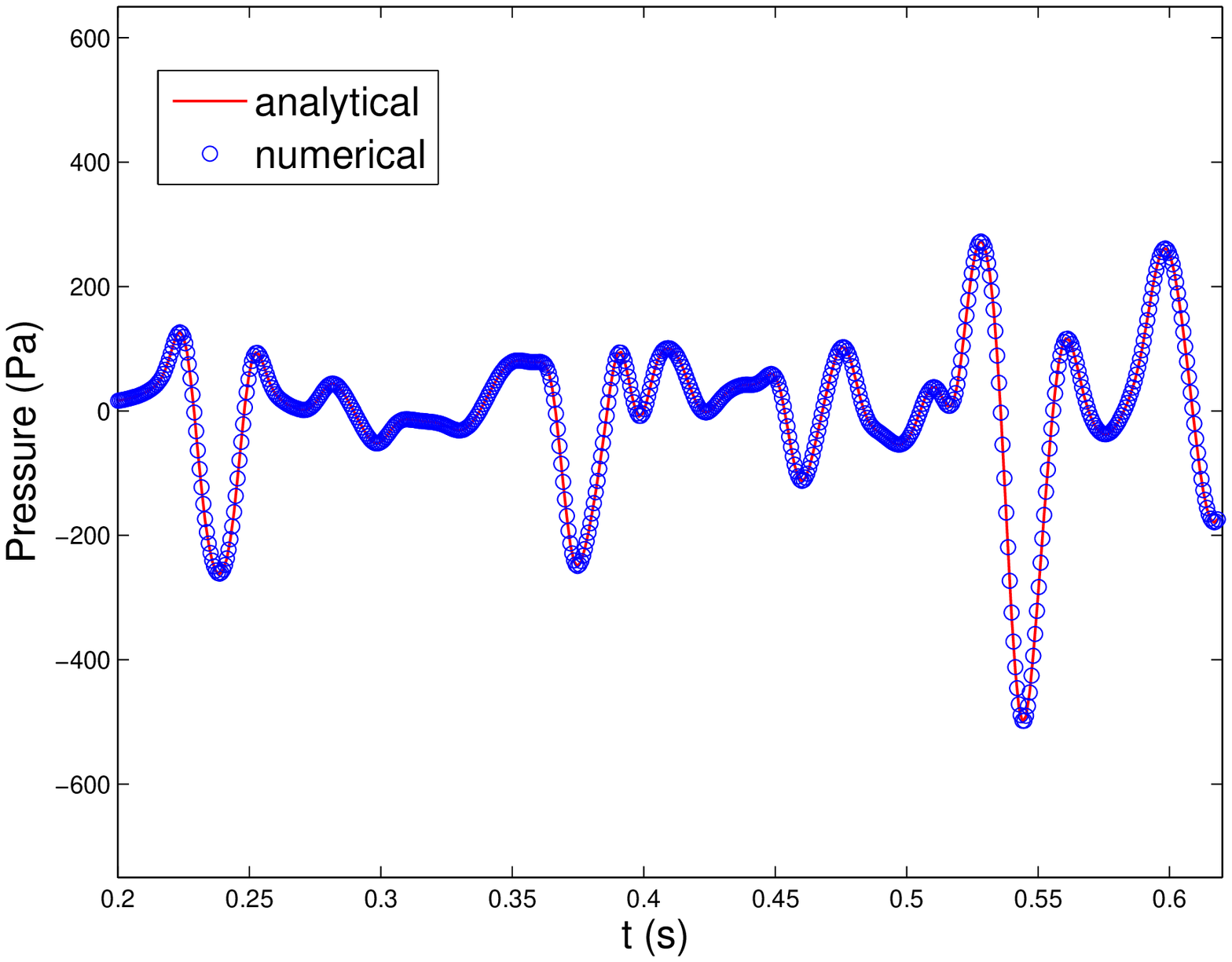}
\end{tabular}
\end{center}
\caption{test 2, inviscid saturating fluid ($\eta=0$ Pa.s). Snapshots of $p$ after 430 time steps (left row), time history of $p$ at the receiver R0 for $t>0.2$ s until $t=0.618$ s corresponding to 900 time steps (right row).}
\label{FigTest2}
\end{figure}

The first experiments concerns the inviscid case ($\eta=0$), where the poroelastic waves are non dispersive. Based on the criterion (\ref{AMR}) and on table \ref{TabParametres}, the refinement factor in the local grid is set to $q=3\approx c_{pf}^\infty\,/\,c_{ps}^\infty$. The numerical values of $p$ after 430 time steps ($t\simeq 0.22$ s) are displayed in the left row of figure \ref{FigTest2}, for the three types of interface conditions. The pressure recorded during 900 time steps at R0 $(x_r=-120,\,y_r=0)$ in $\Omega_0$ is shown in the right row for $t>0.2$ s: the incident wave and the first refracted wave are not represented, in order to focus on the successive reflected/transmitted waves which strongly depend on the interface conditions. This is particularly true for $t>0.4$ s, where shape and amplitude of the recorded pressure completely differ depending on the hydraulic contact.

In the inviscid case, the analytical solutions can be computed very accurately, by a decomposition of Fourier modes on a basis of circular functions. In practice, reference values are obtained by using $N_f=32768$ Fourier modes (with a frequency step $\Delta\,f=0.0063$ Hz) and a truncated basis of 70 Bessel functions. The agreement between numerical and exact values is excellent  in the three cases (figure \ref{FigTest2}, right row). 

\begin{figure}[htbp]
\begin{center}
\begin{tabular}{cc}
(a) & (b)\\
\includegraphics[scale=0.4]{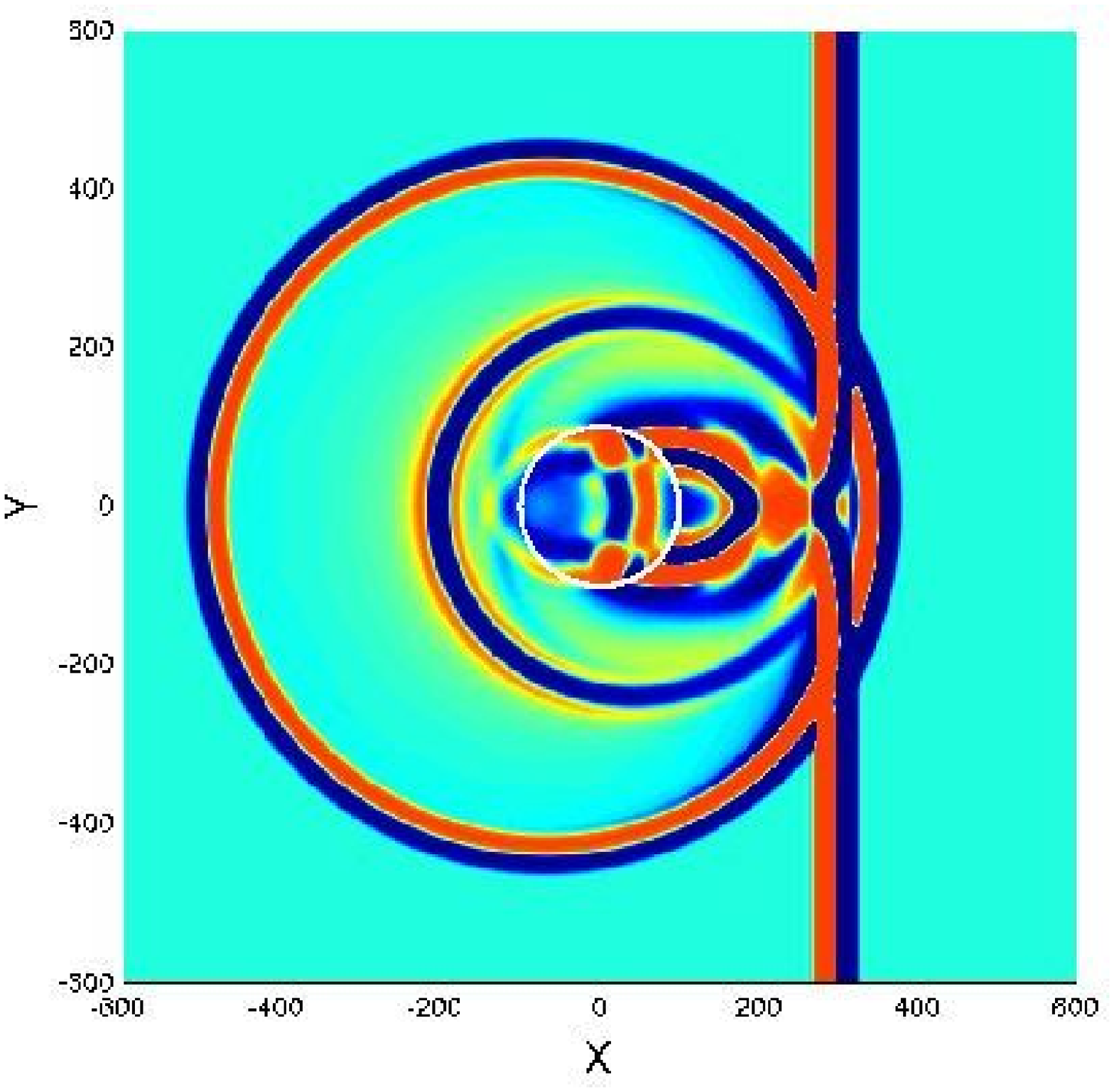}&
\includegraphics[scale=0.4]{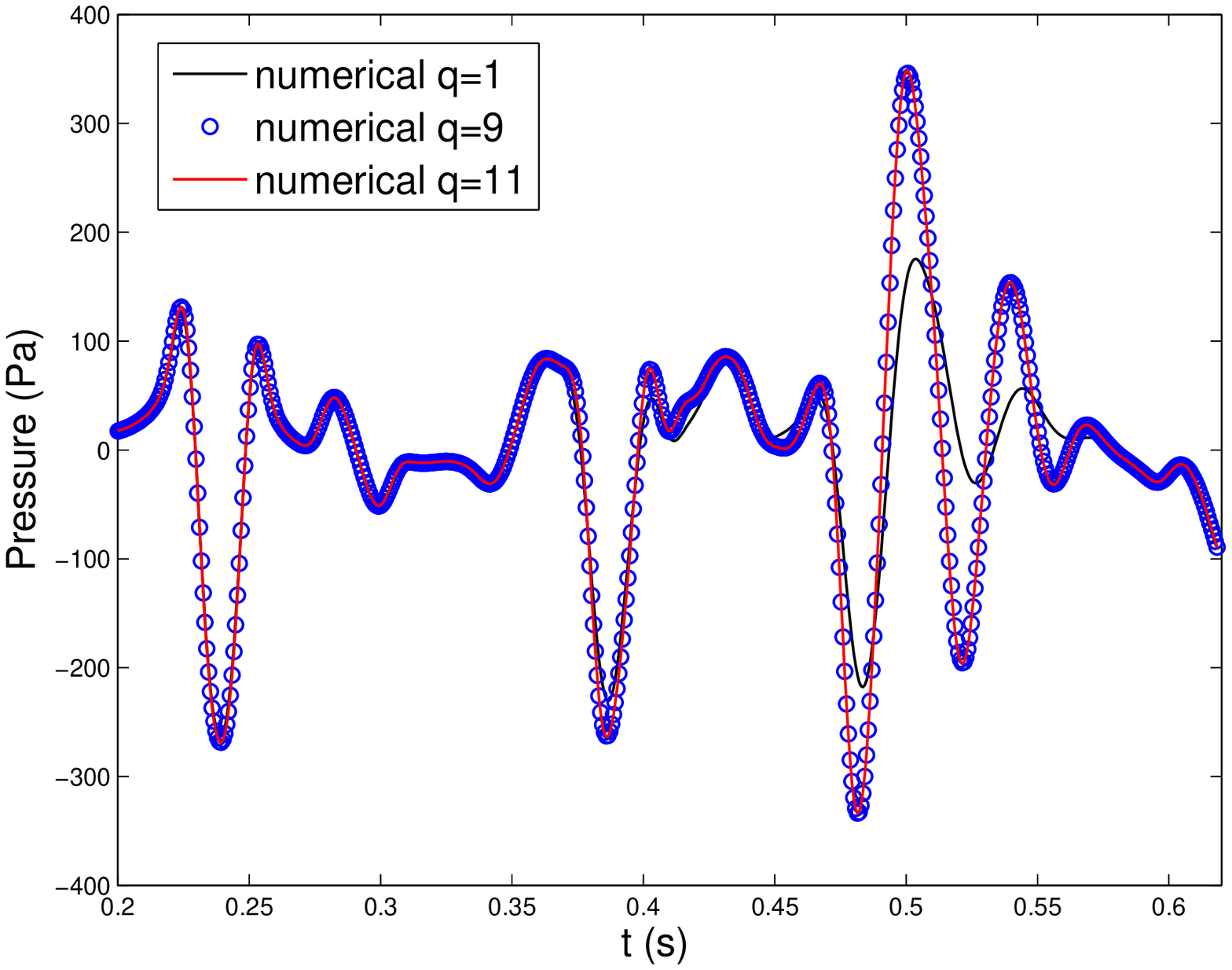}
\end{tabular}
\end{center}
\vspace{-0.8cm}
\caption{test 2, viscous saturating fluid. (a): snapshots of $p$ after 430 iterations for imperfect pore condition. (b): time history of $p$ at the receiver R0 in the acoustic medium, for various refinement factors $q$.}
\label{FigTest2Visq1}
\end{figure}

Similar experiments are also performed with a viscous saturating fluid. The numerical values of $p$ after 430 time steps are displayed in figure \ref{FigTest2Visq1}-a, with imperfect pores. From the criterion (\ref{AMR}) and the phase velocities given in table \ref{TabParametres}, one obtains $q=16$, which is very costly. Nevertheless, the pressure recorded at R0 reveals that $q=9$ suffices to get reference solutions (figure \ref{FigTest2Visq1}-b). 

Compared with the inviscid vase, the viscosity greatly modifies the signal recorded at receiver R0 (figure \ref{FigTest2Visq2}-a). Figure \ref{FigTest2Visq2}-b shows the reflected waves obtained with the three pore conditions. The differences between these signals are smaller than in the inviscid case (right row of figure \ref{FigTest2}).
 
\begin{figure}[htbp]
\begin{center}
\begin{tabular}{cc}
(a) & (b)\\
\includegraphics[scale=0.4]{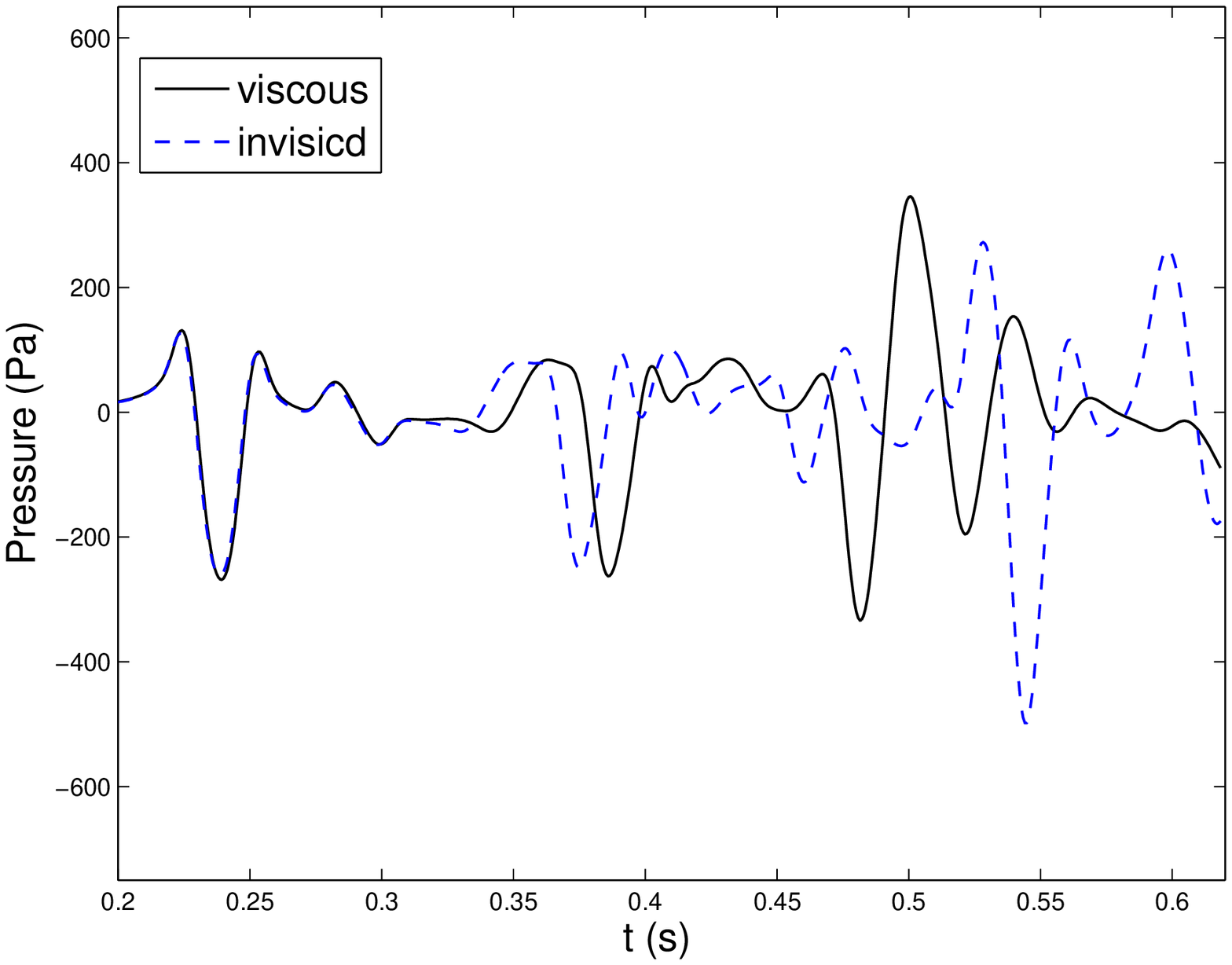}&
\includegraphics[scale=0.4]{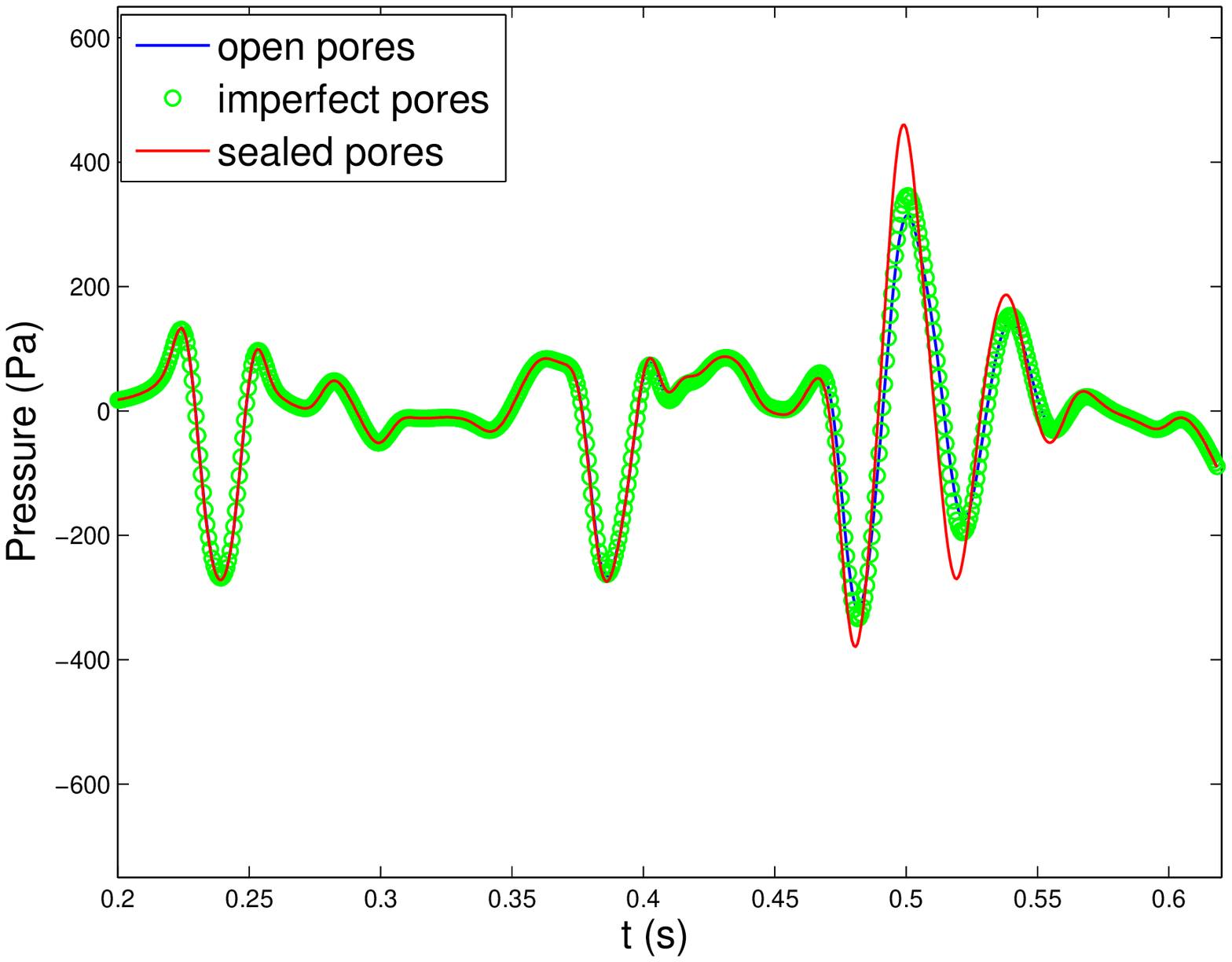}
\end{tabular}
\end{center}
\vspace{-0.8cm}
\caption{test 2, viscous saturating fluid. Time history of $p$. (a): comparison between viscous and inviscid case for imperfect pores; (b): comparison of the three pore conditions.}
\label{FigTest2Visq2}
\end{figure}

%------------------------------------------------------------------------------------------

\subsection{Test 3: a sinusoidal interface}\label{SecResT3}

\begin{figure}[h!]
\begin{center}
\begin{tabular}{cc}
\includegraphics[scale=0.8]{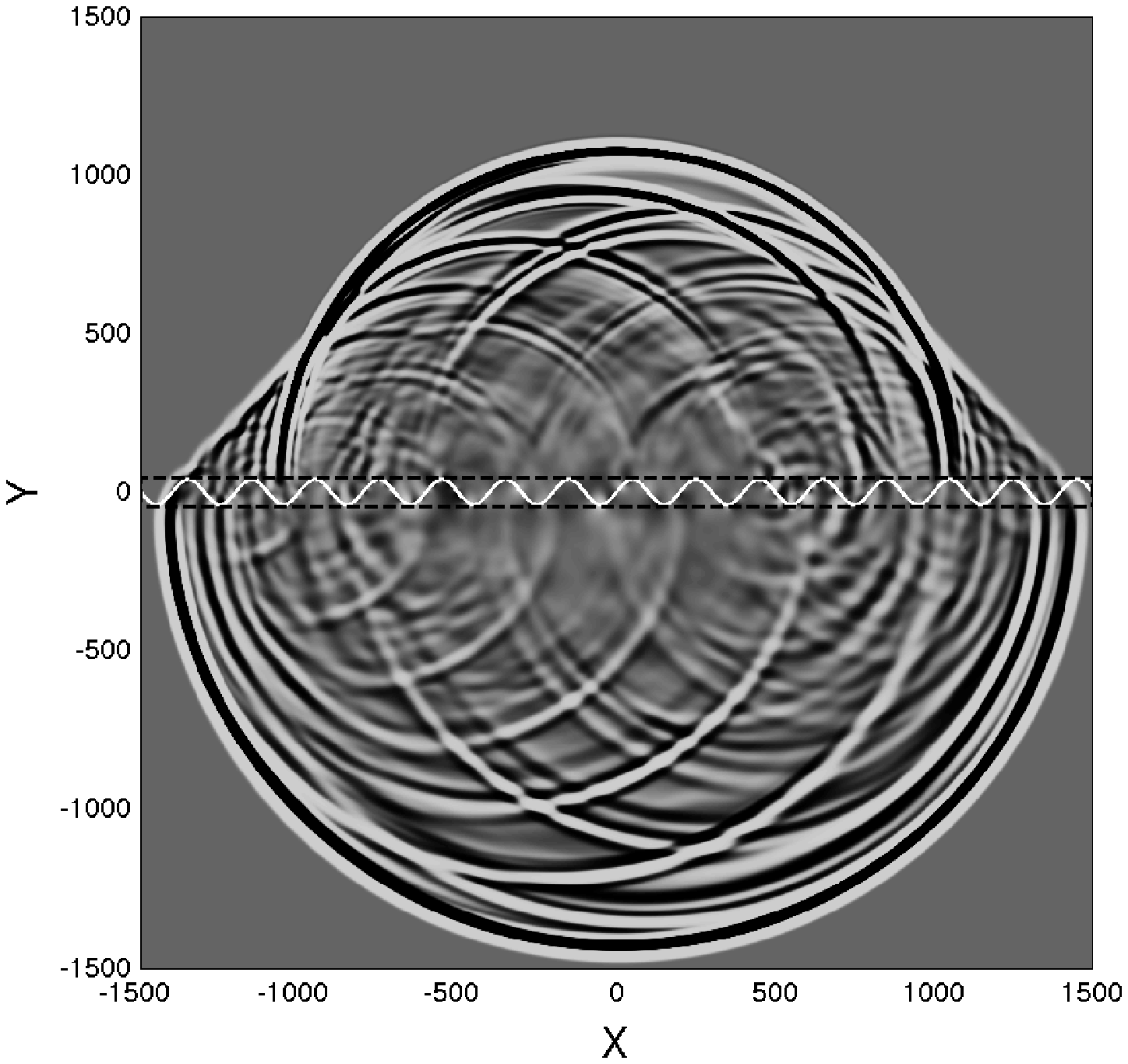}&
\end{tabular}
\end{center}
\vspace{-0.8cm}
\caption{test 3. Snapshot of $p$ after 800 iterations, taking imperfect pore conditions. Horizontal black dotted lines denote the frontiers of the refined grid.}
\label{FigTest3}
\end{figure}

As a third and last test, we consider a sinusoidal interface separating the acoustic medium $\Omega_0$ (top) and the poroelastic medium $\Omega_1$ (bottom). This configuration may model the sea-floor; numerically, it is relevant to test the algorithms with a non-constant curvature of the interface. The domain is $[-1500, 1500]\,\rm{m}^2$, and the interface is given by the relation $y=40 \,\sin\left( \frac{\pi}{100} x \right)$. The viscosity of the saturating fluid is taken into account ($\eta\neq 0$). A source term $f_p$ is put at the point $(x_s=0,\,y_s=20)$ in $\Omega_0$ (\ref{PtSource}). The mesh size is $\Delta\,x=\Delta\,y=2$ m, except in the vicinity of the interface, where a refinement factor of $q=7$ is applied. The refined grid contains about $3.5\,10^6$ nodes and 75544 irregular nodes where the immersed interface method is applied (section \ref{SecNumIIM}).

\begin{figure}[htbp]
\begin{center}
\begin{tabular}{cc}
(a) & (b)\\
\includegraphics[scale=0.4]{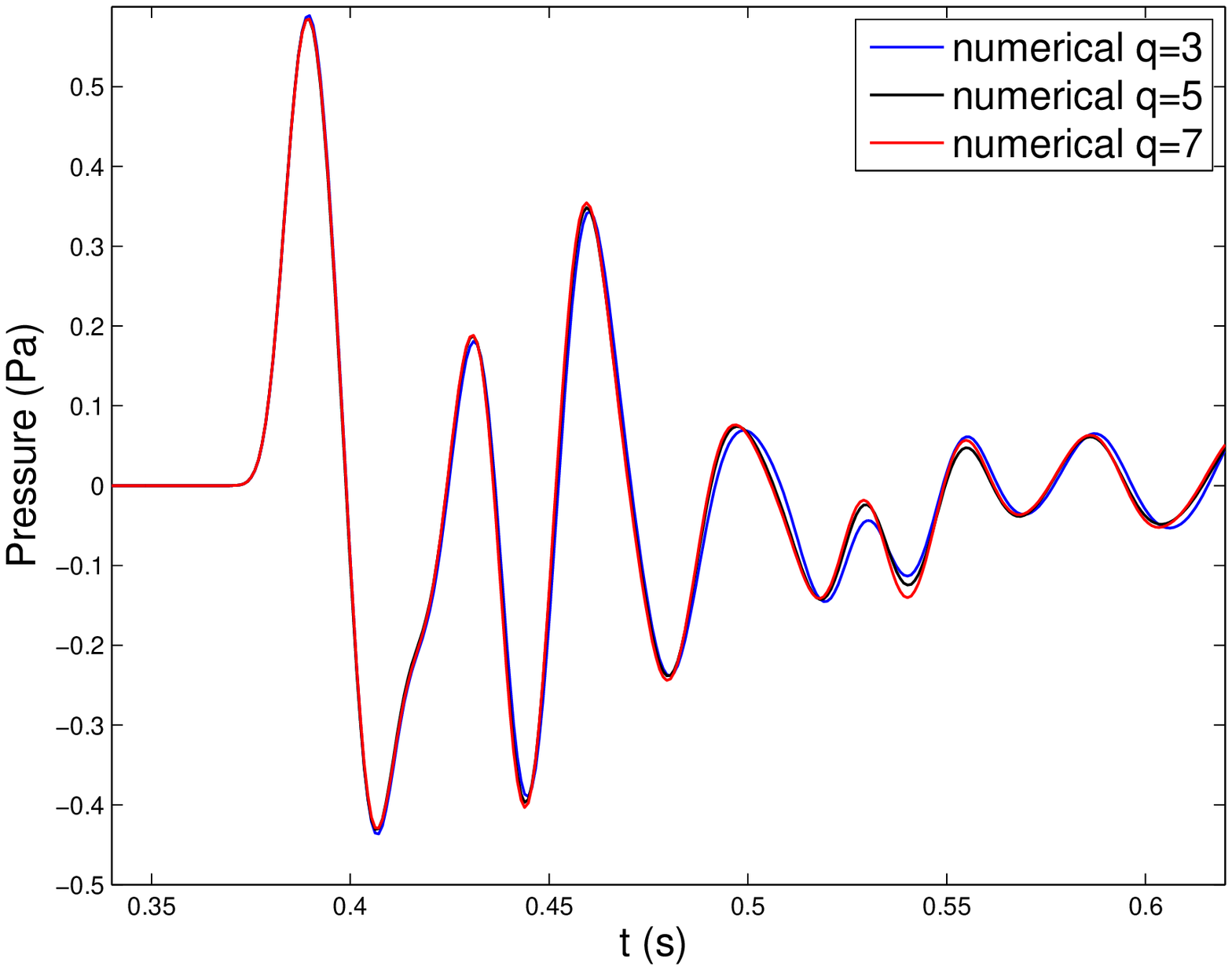}&
\includegraphics[scale=0.4]{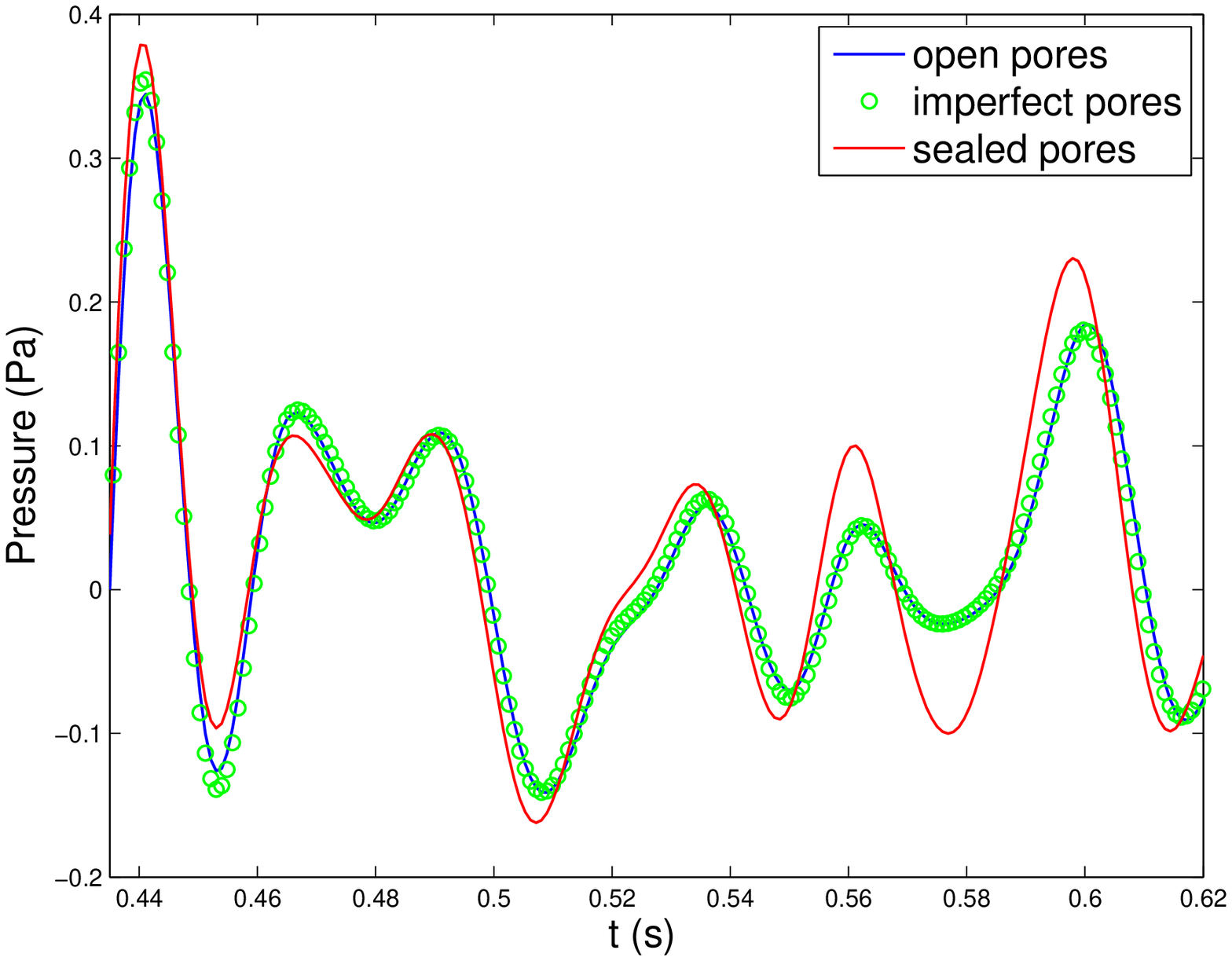}
\end{tabular}
\end{center}
\vspace{-0.8cm}
\caption{test 3. (a): time history of $p$ at the receiver R1 in $\Omega_1$ for different refinement factors; time history of $p$ at the receiver R0 in $\Omega_0$, for $t>0.435$ s (b).}
\label{FigTest3Compar}
\end{figure}

The pressure field after 800 iterations ($t\simeq0.73$ s) is displayed in figure \ref{FigTest3}. From the criterion (\ref{AMR}) and the phase velocities given in table \ref{TabParametres}, one obtains $q=16$, which is very costly. Nevertheless, the pressure recorded close to the interface, at R1 $(x=750,\,y=-45.2)$ in $\Omega_1$, reveals that grid convergence is satisfying when $q=7$ (figure \ref{FigTest3Compar}-a). 

The time history of the pressure recorded at R0 $(x=500,\,y=200)$ in $\Omega_0$ is displayed in figure \ref{FigTest3Compar}-b for the three pore conditions. As in test 2 (figure \ref{FigTest2Visq2}-b), the first reflected waves almost do not depend on the hydraulic contact. Consequently, we focus on the subsequent waves ($t>0.435$ s), where differences are clearly observed depending on the interface conditions. 

%------------------------------------------------------------------------------------------

\section{Conclusion}\label{SecConclu}

We have developed a robust and highly accurate numerical model to simulate wave propagation in fluid / poroelastic media. Our model can incorporate various models of interface conditions, in particular the case of imperfect hydraulic contact: to our knowledge, it is the first time that such simulations are proposed. Arbitrary-shaped interfaces can be handled and accuracy is ensured by a subcell resolution on a Cartesian grid.

Numerical experiments have shown that each part of the algorithm is required to get efficiently reliable results: fourth-order scheme with time splitting, mesh refinement, immersed interface method. The effect of interface conditions on the diffracted waves has been illustrated. When the viscous effects are noticeable, the accurate computation of the slow compressional wave in the poroelastic medium is crucial for the overall accuracy. This diffusive wave propagates along the interface and plays a major role in the balance equations. 

This numerical model enables to investigate many physically-relevant configurations. For instance, comparisons between real experiments and simulations could be used to characterize the hydraulic permeability $\mathcal{K}$ in (\ref{JC}). Simulation of multiple scattering in random or periodic media is another fruitful application. The objective is to estimate numerically the properties of the homogenized effective medium \cite{LUPPE08}. Since a Cartesian grid is used and no meshing of the interfaces is required, our approach is very well suited to the modeling of numerous scatterers (typically, a few hundreds \cite{CHEKROUN09}).

Our numerical model is valid only in the low-frequency range. For frequencies greater than $f_c$ in (\ref{Fc}), the second equation of (\ref{LCBiot}) must be modified to account for the viscous layer dissipation. Fractional derivatives of order 1/2 are then introduced \cite{LU05}. We are currently investigating this topic \cite{BLANC12}. Similar extensions are also required concerning the hydraulic permeability in (\ref{JC}). Indeed, it is known that $\mathcal{K}$ depends not only on the geometrical properties of the medium, but also on the frequency \cite{ROSENBAUM74}.

%------------------------------------------------------------------------------------------

\appendix

\section{Matrices of interface conditions}\label{SecAnnexeMat}

The matrices ${\bf C}_0^0$, ${\bf C}_1^0$ and ${\bf L}_1^0$ introduced in (\ref{JC0}) are detailed. In the case of open pore conditions (\ref{JCOpen}), it follows from (\ref{NT}) 
\begin{equation}
\begin{array}{l}
{\bf C}_0^0(\tau)=
\left(
\begin{array}{ccc}
y^{'} & - x^{'} & 0\\
0     & 0       & -\left(x^{'2}+y^{'2}\right)
\end{array}
\right),\\
[12pt]
{\bf C}_1^0(\tau)=
\left(
\begin{array}{cccccccc}
y^{'} & - x^{'} & y^{'} & - x^{'} & 0      & 0                & 0      & 0 \\
0     & 0       & 0     & 0       & y^{'2} & -2\,x^{'}\,y^{'} & x^{'2} & 0
\end{array}
\right),\\
[12pt]
{\bf L}_1^0(\tau)=
\left(
\begin{array}{cccccccc}
0 & 0 & 0 & 0 & x^{'}\,y^{'} & y^{'2}-x^{'2} & - x^{'}\,y^{'} & 0 \\
0 & 0 & 0 & 0 & y^{'2} & -2\,x^{'}\,y^{'} & x^{'2} & x^{'2}+y^{'2}
\end{array}
\right).
\end{array}
\label{MatOpen}
\end{equation}
In the case of sealed pore conditions (\ref{JCClose}), one gets
\begin{equation}
\begin{array}{l}
{\bf C}_0^0(\tau)=
\left(
\begin{array}{ccc}
y^{'} & - x^{'} & 0\\
0     & 0       & -\left(x^{'2}+y^{'2}\right)
\end{array}
\right),\\
[12pt]
{\bf C}_1^0(\tau)=
\left(
\begin{array}{cccccccc}
y^{'} & - x^{'} & 0 & 0 & 0      & 0                & 0      & 0 \\
0     & 0       & 0 & 0 & y^{'2} & -2\,x^{'}\,y^{'} & x^{'2} & 0
\end{array}
\right),\\
[12pt]
{\bf L}_1^0(\tau)=
\left(
\begin{array}{cccccccc}
0 & 0 & 0     & 0      & x^{'}\,y^{'} & y^{'2}-x^{'2} & - x^{'}\,y^{'} & 0 \\
0 & 0 & y^{'} & -x^{'} & 0            & 0             & 0              & 0 
\end{array}
\right).
\end{array}
\label{MatClose}
\end{equation}
In the case of imperfect pore conditions (\ref{JCPerm}), one gets
\begin{equation}
\begin{array}{l}
{\bf C}_0^0(\tau)=
\left(
\begin{array}{ccc}
y^{'} & - x^{'} & 0\\
0     & 0       & -\left(x^{'2}+y^{'2}\right)
\end{array}
\right),\\
[12pt]
{\bf C}_1^0(\tau)=
\left(
\begin{array}{cccccccc}
y^{'} & - x^{'} & y^{'} & - x^{'} & 0      & 0                & 0      & 0 \\
0     & 0       & 0     & 0       & y^{'2} & -2\,x^{'}\,y^{'} & x^{'2} & 0
\end{array}
\right),\\
[12pt]
{\bf L}_1^0(\tau)=
\left(
\begin{array}{cccccccc}
0 & 0 & 0 & 0 & x^{'}\,y^{'} & y^{'2}-x^{'2} & - x^{'}\,y^{'} & 0 \\
0 & 0 & \displaystyle \frac{\textstyle y^{'}}{\textstyle \mathcal{K}}\sqrt{x^{'2}+y^{'2}} & \displaystyle -\frac{\textstyle x^{'}}{\textstyle \mathcal{K}}\sqrt{x^{'2}+y^{'2}} & y^{'2} & -2\,x^{'}\,y^{'} & x^{'2} & x^{'2}+y^{'2}
\end{array}
\right).
\end{array}
\label{MatPerm}
\end{equation}

%------------------------------------------------------------------------------------------

\section{Four steps to build ${\bf U}_{I,J}^*$ (\ref{UIJ*})}\label{SecDetailsIIM}

{\bf Step 1: high-order interface conditions}. The conditions (\ref{JC0}) are differentiated in terms of $t$ and $\tau$. Time derivatives are replaced by spatial derivatives thanks to the splitted evolution equations (equation (\ref{Strang1})). Consider for instance the zero-th order boundary conditions ${\bf L}_1^0\,{\bf U}_1^0={\bf 0}$; time derivative yields
\begin{equation}
\frac{\textstyle \partial}{\textstyle \partial\,t}\,{\bf L}_1^0\,{\bf U}_1^0={\bf L}_1^0\,\frac{\textstyle \partial}{\textstyle \partial\,t}\,{\bf U}_1^0={\bf L}_1^0\,\left(-{\bf A}_1\,\frac{\textstyle \partial}{\textstyle \partial\,x}\,{\bf U}_1^0-{\bf B}_1\,\frac{\textstyle \partial}{\textstyle \partial\,y}\,{\bf U}_1^0\right)={\bf 0}.
\label{DerT}
\end{equation}
Derivatives in terms of $\tau$ are replaced by spatial derivatives thanks to the chain-rule. For instance, ${\bf L}_1^0\,{\bf U}_1^0={\bf 0}$ yields
\begin{equation}
\frac{\textstyle \partial}{\textstyle \partial\,\tau}\,{\bf L}_1\,{\bf U}_1^0=\left(\frac{\textstyle \partial}{\textstyle \partial\,\tau}\,{\bf L}_1^0 \right)\,{\bf U}_1^0+{\bf L}_1^0\,\left(x^{'}\,\frac{\textstyle \partial}{\textstyle \partial\,x}\,{\bf U}_1^0+y^{'}\,\frac{\textstyle \partial}{\textstyle \partial\,y}\,{\bf U}_1^0\right)={\bf 0}.
\label{DerTau}
\end{equation}
From (\ref{JC0}), (\ref{DerT}) and (\ref{DerTau}), we build a matrix ${\bf L}_1^1$ such that ${\bf L}_1^1\,{\bf U}_1^1=0$, which provides first-order boundary conditions. A similar procedure is applied to the zeroth-th order jump conditions. By iterating this process $r$ times, $r$-th order interface conditions are obtained
\begin{equation}
{\bf C}_1^r\,{\bf U}_1^r={\bf C}_0^r\,{\bf U}_0^r,\qquad
{\bf L}_1^r\,{\bf U}_1^r={\bf 0}.
\label{JCr}
\end{equation}
In our codes, the matrices ${\bf C}_i^r$ and ${\bf L}_1^r$ are computed automatically by developing computer algebra tools.\\

\noindent
{\bf Step 2: high-order compatibility conditions}. Some components of the successive spatial derivatives of ${\bf U}$ are not independent. In the fluid medium $\Omega_0$, the vorticity of acoustic velocity is null ${\bf \nabla}\wedge {\bf v}={\bf 0}$, which yields
\begin{equation}
\frac{\textstyle \partial\,v_2}{\textstyle \partial\,x}=\frac{\textstyle \partial\,v_1}{\textstyle \partial\,y}.
\label{RotV}
\end{equation}
In the porous medium $\Omega_1$, the symmetry of $\sigma$ yields the Beltrami-Michell equation
\begin{equation}
\begin{array}{l}
\displaystyle
\frac{\textstyle \partial^2 \,\sigma_{12}}{\textstyle \partial \,x\,\partial\,y}
=\theta_0\,\frac{\textstyle \partial^2 \,\sigma_{11}}{\textstyle \partial \,x^2}
+\theta_1\,\frac{\textstyle \partial^2 \,\sigma_{22}}{\textstyle \partial \,x^2}
+\theta_2\,\frac{\textstyle \partial^2 \,p}{\textstyle \partial \,x^2}
+\theta_1\,\frac{\textstyle \partial^2 \,\sigma_{11}}{\textstyle \partial \,y^2}
+\theta_0\,\frac{\textstyle \partial^2 \,\sigma_{22}}{\textstyle \partial \,y^2}
+\theta_2\,\frac{\textstyle \partial^2 \,p}{\textstyle \partial \,y^2},\\
[10pt]
\displaystyle
\theta_0=-\frac{\textstyle \lambda_0}{\textstyle 4\,(\lambda_0+\mu)},\quad \theta_1=\frac{\textstyle \lambda_0+2\,\mu}{\textstyle 4\,(\lambda_0+\mu)},\quad \theta_2=\frac{\textstyle \mu \,\beta}{\textstyle 2\,(\lambda_0+\mu)}.
\end{array}
\label{Beltrami}
\end{equation}
Equations (\ref{RotV}) and (\ref{Beltrami}) can be differentiated in terms of $x$ and $y$ as many times as required, assuming a sufficiently smooth solution. The equations so-obtained are satisfied at any point, in particular on both sides of $\Gamma$. They can be used to reduce the number of components of ${\bf U}_i^r$. Reduced vectors ${\bf V}_i$ are defined, such that
\begin{equation}
{\bf U}_i^r = {\bf G}_i^r\,{\bf V}_i^r, 
\label{UGV}
\end{equation}
where ${\bf G}_i^r$ are $n_v \times (n_v-n_b)$ matrices. In $\Omega_0$, $n_b=r\,(r+1)/2$ if $r\geq 1$, 0 else; an algorithm to compute ${\bf G}_0^r$ is given in appendix A of \cite{LOMBARD04}. In $\Omega_1$, $n_b=r\,(r-1)/2$ if $r\geq 2$, $n_b=0$ else; an algorithm to compute ${\bf G}_1^r$ is given in appendix B of \cite{CHIAVASSA11}. The compatibility conditions are crucial to ensure the stability of the immersed interface method. \\ 

\noindent
{\bf Step 3: high-order boundary values}. On the poroelastic $\Omega_1$ side, the $r$-th order boundary conditions in (\ref{JCr}) and the high-order Beltrami equations (\ref{UGV}) yield the underdetermined linear system
\begin{equation}
{\bf L}_1^r\,{\bf G}_1^r\,{\bf V}_1^r={\bf 0}.
\end{equation}
It leads to
\begin{equation}
{\bf V}_1^r={\bf K}_1^r\,{\bf W}_1^r,
\label{VKW}
\end{equation}
where ${\bf K}_1^r$ is the matrix built from the kernel of ${\bf L}_1^r\,{\bf G}_1^r$. The solution ${\bf W}_1^r$ is the minimum set of independent components of the trace of ${\bf U}$ and its spatial derivatives up to the $r$-th order, on the domain $\Omega_1$. Injecting (\ref{VKW}) into the $r$-th order jump conditions (\ref{JCr}) gives
\begin{equation}
{\bf S}_1^r\,{\bf W}_1^r={\bf S}_0^r\,{\bf V}_0^r ,
\label{CGK}
\end{equation}
where ${\bf S}_1^r={\bf C}_1^r\,{\bf G}_1^r\,{\bf K}_1^r$ and ${\bf S}_0^r={\bf C}_0^r\,{\bf G}_0^r$. The underdetermined system (\ref{CGK}) is solved
\begin{equation}
{\bf W}_1^r=\left(\left({\bf S}_1^r\right)^{-1}\,{\bf S}_0^r\,|\,{\bf R}_{{\bf S}_1^r}\right)
\left(
\begin{array}{c}
\displaystyle
{\bf V}_0^r\\
[8pt]
\displaystyle
{\bf \Lambda}^r
\end{array}
\right),
\label{SVD}
\end{equation}
where $({\bf S}_1^r)^{-1}$ is the least-squares pseudo-inverse of ${\bf S}_1^r$, ${\bf R}_{{\bf S}_1^r}$ is the matrix containing the kernel of ${\bf S}_1^r$, and ${\bf \Lambda}^r$ is a set of Lagrange multipliers. To build $({\bf S}_1^r)^{-1}$ and ${\bf R}_{{\bf S}_1^r}$, a singular value decomposition of ${\bf S}_1^r$ is performed.\\

\noindent
{\bf Step 4: construction of modified values}. The coefficients of 2-D Taylor expansions around $P$  (figure \ref{Patate}) are put in the matrix ${\bf \Pi}_{i,j}^r$: 
\begin{equation}
{\bf \Pi}_{i,j}^r=\left({\bf I},...,
\frac{\textstyle 1}{\textstyle \beta \,!\,(\alpha-\beta)\,!}\,(x_i-x_P)^{\alpha-\beta}(y_j-y_P)^\beta\,{\bf I},...,
\frac{\textstyle (y_j-y_P)^r}{\textstyle r\,!}\,{\bf I}
\right),
\label{Taylor}
\end{equation}
where $\alpha=0,...,\,r$ and $\beta=0,...,\,\alpha$; ${\bf I}$ is the $3\times3$ or $8\times8$ identity matrix, depending on whether $(x_i,\,y_j)$ belongs to the fluid domain $\Omega_0$ or the poroelastic domain $\Omega_1$. By definition, the modified value at $(x_I,\,y_J)$ is
\begin{equation}
{\bf U}_{I,J}^*={\bf \Pi}_{I,J}^r\,{\bf U}_0^r.
\label{SolModa}
\end{equation}
To determine the trace ${\bf U}_0^r$ in (\ref{SolModa}), consider the disc ${\cal D}$ centered at $P$ with radius $d$ (figure \ref{Patate}). At the grid nodes of ${\cal D}\cap\Omega_0$, $r$-th order Taylor expansions at $P$, (\ref{UGV}) and (\ref{VKW}) give
\begin{equation}
\begin{array}{lll}
\overline{\bf U}_{i,j}^{(1)} &=& {\bf \Pi}_{i,j}^r\,{\bf U}_0^r,\\
[4pt]
&=& {\bf \Pi}_{i,j}^r\,{\bf G}_0^r\,{\bf V}_0^r,\\
[4pt]
&=& {\bf \Pi}_{i,j}^r\,{\bf G}_0^r\,\left({\bf 1}\,|\,{\bf 0}\right)
\left(
\begin{array}{c}
{\bf V}_0^r\\
[4pt]
{\bf \Lambda}^r
\end{array}
\right),
\end{array}
\label{Taylor1}
\end{equation}
where Taylor rests have been omitted for the sake of clarity. At the grid nodes of ${\cal D}\cap\Omega_1$, $r$-th order Taylor expansions at $P$, (\ref{UGV}), (\ref{VKW}) and (\ref{SVD}), give 
\begin{equation}
\begin{array}{lll}
\overline{\bf U}_{i,j}^{(1)} &=& {\bf \Pi}_{i,j}^r\,{\bf U}_1^r,\\
[4pt]
&=& {\bf \Pi}_{i,j}^r\,{\bf G}_1^r\,{\bf K}_1^r\,{\bf W}_1^r,\\
[4pt]
&=& {\bf \Pi}_{i,j}^r\,{\bf G}_1^r\,{\bf K}_1^r\,\left(\left({\bf S}_1^r\right)^{-1}\,|\,{\bf R}_{{\bf S}_1^r}\right)
\left(
\begin{array}{c}
{\bf V}_0^r\\
[4pt]
{\bf \Lambda}^r
\end{array}
\right).
\end{array}
\label{Taylor2}
\end{equation}
The equations (\ref{Taylor1}) and (\ref{Taylor2}) are written using an adequate matrix ${\bf M}$
\begin{equation}
\left(
{\bf U}^{(1)}
\right)_{\mathcal D}
={\bf M}
\left(
\begin{array}{c}
{\bf V}_0^r\\
[4pt]
{\bf \Lambda}^r
\end{array}
\right).
\label{Taylor3}
\end{equation}
The least-squares inverse of the matrix ${\bf M}$ is denoted by ${\bf M}^{-1}$. Since the Lagrange multipliers ${\bf \Lambda}^r$ are not involved in (\ref{SolModa}), ${\bf M}^{-1}$ is restricted to $\overline{\bf M}^{-1}$, so that
\begin{equation}
{\bf V}_0^r=\overline{\bf M}^{-1}
\,\left(
{\bf U}^{(1)}
\right)_{\mathcal D}.
\label{Taylor4}
\end{equation}
The modified value follows from (\ref{UGV}), (\ref{VKW}), (\ref{SolModa}) and (\ref{Taylor4}), recovering (\ref{UIJ*}):
\begin{equation}
\begin{array}{lll}
\displaystyle
{\bf U}_{I,J}^* &=& \displaystyle {\bf \Pi}_{I,J}^r\,{\bf G}_0^r\,{\bf K}_0^r\,\overline{\bf M}^{-1}\,\left({\bf U}^{(1)}\right)_{\mathcal D},\\
[12pt]
&=& \displaystyle {\cal M}\,\left({\bf U}^{(1)}\right)_{\mathcal D}.
\end{array}
\label{Esim}
\end{equation}

%------------------------------------------------------------------------------------------

\section{Normalization of the variables}\label{SecAnnexeNorme}

As mentioned in section \ref{SecNumIIM}, normalized parameters and unknowns are used in our computer codes. These quantities are denoted by tildes in the following. Given a real ${\cal N}$, we define the normalized time
\begin{equation}
\tilde{t}= {\cal N}\,t,
\label{NormeT}
\end{equation}
the normalized variables
\begin{equation}
\tilde{{\bf v}}={\cal N}\,{\bf v},\quad \tilde{{\bf v}}_s={\cal N}\,{\bf v}_s,\quad \tilde{{\bf w}}={\cal N}\,{\bf w},\quad \tilde{\sigma}=\frac{\textstyle \sigma}{\textstyle {\cal N}},\quad \tilde{p}=\frac{\textstyle p}{\textstyle {\cal N}},
\label{NormeVar}
\end{equation}
and the normalized physical parameters
\begin{equation}
\begin{array}{l}
\displaystyle
\tilde{\rho_f}=\frac{\textstyle \rho_f}{\textstyle {\cal N}},\quad \tilde{\rho_s}=\frac{\textstyle \rho_s}{\textstyle {\cal N}},\quad \tilde{\rho_w}=\frac{\textstyle \rho_w}{\textstyle {\cal N}},\\
[10pt]
\displaystyle
\tilde{\lambda_f}=\frac{\textstyle \lambda_f}{\textstyle {\cal N}^3},\quad \tilde{\mu}=\frac{\textstyle \mu}{\textstyle {\cal N}^3},\quad \tilde{m}=\frac{\textstyle m}{\textstyle {\cal N}^3},\\
[10pt]
\displaystyle
\tilde{\eta}=\frac{\textstyle \eta}{\textstyle {\cal N}},\quad \tilde{\kappa}={\cal N}\,\kappa,\quad\tilde{\mathcal{K}}={\cal N}^2\,\mathcal{K}.
\end{array}
\label{NormeParam}
\end{equation}
The value of the normalization parameter is set to ${\cal N}=1000$ in numerical experiments.

%------------------------------------------------------------------------------------------

\end{document}